\documentclass[12pt,technote, onecolumn, draft]{IEEEtran}
\usepackage{latexsym}
\usepackage{mathrsfs}
\usepackage{graphicx}
\usepackage{multirow}
\usepackage{amsfonts,amssymb,amsmath,amsthm,bm}
\usepackage{color}
\usepackage{cite}
\newtheorem{theorem}{Theorem}[section] 
\newtheorem{definition}[theorem]{Definition} 
\newtheorem{lemma}[theorem]{Lemma} 
\newtheorem{corollary}[theorem]{Corollary}
\newtheorem{example}[theorem]{Example}
\newtheorem{proposition}[theorem]{Proposition}
\newtheorem{remark}[theorem]{Remark}

\begin{document}

\title{On a Class of Twisted Elliptic Curve Codes$^{\dag}$}
\author{Xiaofeng Liu, Jun Zhang, Fang-Wei Fu
\IEEEcompsocitemizethanks{\IEEEcompsocthanksitem Xiaofeng Liu and Fang-Wei Fu are with the Chern Institute of Mathematics and LPMC, Nankai University, Tianjin 300071, China, Emails: lxfhah@mail.nankai.edu.cn, fwfu@nankai.edu.cn. Jun Zhang is with the School of Mathematical Sciences, Capital Normal University, Beijing 100048, China, Email: junz@cnu.edu.cn.
}
\thanks{$^\dag$Xiaofeng Liu and Fang-Wei Fu were supported by the National Key Research and Development Program of China (Grant Nos.  2022YFA1005000), the National Natural Science Foundation of China (Grant Nos. 12141108, 61971243), the Fundamental Research Funds for the Central Universities of China (Nankai University), and the Nankai Zhide Foundation. Jun Zhang was supported by the National Natural Science Foundation of China under Grant Nos. 12222113, 12441105.} 
}{\tiny }

\maketitle

\begin{abstract} 
 Motivated by the research of twisted generalized Reed-Solomon (TGRS) codes, we initiate the study of twisted elliptic curve codes (TECCs) in this paper. In particular, we study a class of TECCs with one twist. The parity-check matrices of the TECCs are explicitly given by computing the Weil differentials. Then the sufficient and necessary conditions of self-duality are presented. The minimum distances of the TECCs are also determined. Moreover, examples of MDS, AMDS, self-dual, and MDS self-dual TECCs are given. Finally, we calculate the dimensions of the Schur squares of the TECCs. Then we obtained some nonequivalence results.
\end{abstract}


\section{Introduction}
\label{sec:1}
Let $\mathbb{F}_{q}$ be a finite field with $q$ an odd prime power and $\mathbb{F}_{q}^{*}=\mathbb{F}_{q}\setminus\{0\}$. A linear code $\mathcal{C}$ with parameters $[n, k, d]_{q}$ is a subspace of $\mathbb{F}^{n}_{q}$. The Singleton bound $d\leq n-k+1$ is the most famous trade-off of the parameters $n, k, d$. Any linear code that achieves the Singleton bound is called a maximum distance separable (MDS) code. The number $\mathcal{S}(\mathcal{C})=n-k+1-d$ is called the Singleton defect of the code $\mathcal{C}$. If $\mathcal{S}(\mathcal{C})=0$, then $\mathcal{C}$ is just an MDS code. If $\mathcal{S}(\mathcal{C})=1$, then $\mathcal{C}$ is called an almost-MDS (AMDS) code. If $\mathcal{S}(\mathcal{C})=\mathcal{S}(\mathcal{C}^{\perp})=1$, then $\mathcal{C}$ is called a near-MDS (NMDS) code.
	
	The Generalized Reed-Solomon (GRS) code is one of the most important MDS code families. In 2017, Beelen \textit{et al.} introduced the concept of twisted generalized Reed-Solomon (TGRS) codes; see \cite{4}. Since then, properties of the TGRS codes have attracted a lot of researchers. Now we recall some classical constructions of TGRS codes. In \cite{18}, Sun \textit{et al.} studied the decoding algorithms of TGRS codes and twisted Goppa codes. Huang \textit{et al.} constructed MDS or NMDS LCD codes from twisted genralized Reed-Solomon codes; see \cite{5}.
	In \cite{22}, Hu \textit{et al.} came up with a new class of TGRS codes, namely $(\mathcal{L},\mathcal{P})-$TGRS codes and provided necessary and sufficient conditions for $(\mathcal{L},\mathcal{P})-$TGRS codes to be MDS and self-dual, which extends the related results in the previous work about TGRS codes.  Recently, the deep hole problems of TGRS codes were also considered; see~\cite{23,24} \textit{etc.}.

	Algebraic geometry (AG) codes as generalizations of RS codes were introduced by Goppa in the 1980s. They are widely applied in both coding theory and cryptography. The Tasfasman-Vladut-Zink bound induced by the towers of AG codes improves the Gilbert-Varshamov bound in the asymptotic performance of codes; see~\cite{1}. The cryptanalysis of the McEliece cryptosystem based on AG codes and their subfield subcodes can be checked at~\cite{21} \textit{etc.}. The AG codes constructed on the elliptic curves over finite fields are called elliptic curve codes (ECCs).  ECCs are widely applied in cryptography and coding theory.  Identifying the minimum distance of ECCs is equivalent to a subset-sum problem (SSP) and  it is $\mathbf{NP}-$hard under $\mathbf{RP}-$reduction (see \cite{2}). Now we recall some results of ECCs. Deterministic results on the minimum distance of the ECCs were obtained; see~\cite{8,13,20}. Recently, covering radius problems and the deep hole problems of ECCs were considered in~\cite{2}. In~\cite{3}, Chen calculated the Schur square of the ECCs and constructed many non-Reed-Solomon-type MDS codes.   
	In \cite{11}, Thiranant \textit{et al.} compared the performance between QR code authentication based on RSA and elliptic curve cryptography. In \cite{12},  Li \textit{et. al.} used ECCs to construct optimal locally repairable codes. Genç and Afacan designed and implemented an efficient elliptic curve digital signature algorithm (ECDSA); see \cite{19}. 
	\subsection{Main Techniques and Results}
	Let $E/\mathbb{F}_{q}$ be an elliptic function field on a finite field $\mathbb{F}_{q}$ with $q$ an odd prime power.  In this paper, we initiate the study of a class of twisted elliptic curve codes (TECCs), which is also the first construction of twisted AG codes. In this paper, we focus on the elliptic function fields over finite fields with odd characteristics. Such function fields belong to Kummer extensions with defining equation $y^{2}=f(x)$ for some square-free polynomial $f(x)$ with $\deg(f(x))=3$. By using the Riemann-Hurwitz formula and prime ideal decompositions in the elliptic function fields, we completely determine the parity-check matrices of ECCs. The orthogonality relations and the parity-check matrices of ECCs are listed in Theorem \ref{cc}  and Theorem \ref{c1}. By analyzing the codimension between TECCs and classical ECCs, we also completely determine the parity-check matrices of TECCs and the results are listed in Theorems \ref{30}. The results of self-duality are listed in Corollary \ref{59}. Conditions such that TECCs attain the possible minimum distances are also obtained and the results are listed in Theorem \ref{10}. Last, we calculate the dimension of the Schur squares of TECCs through which we obtained some non-equivalence results.
	\subsection{Orgnization of this paper}
	In this paper, we extend the classical ECCs to twisted elliptic curve codes (TECCs). In particular, we study a class of TECCs with one twist and determine their dual codes, minimum distances, MDS condition and self-duality. The rest of the paper is organized as follows. In Section \ref*{sec:2}, we introduce some basic results of elliptic function fields and some constructions of ECCs and propose the definition of TECCs. In Section \ref*{sec:4},  by using the Riemann-Hurwitz formula and computing the Weil differentials, we give explicit constructions of the parity-check matrices for the ECC $\mathcal{C}_{\mathcal{L}}(D, kO)$. In Section \ref*{sec:5}, we give some constructions of the parity-check matrices for TECCs and determine the conditions such that TECCs $\mathcal{C}(D, kO,\ell, \eta)$ to be self-dual. In Section \ref*{sec:6}, we determine the conditions such that TECCs attain the possible minimum distances. In Section \ref{sec:8}, by calculating the dimensions of the Schur squares of TECCs, we obtain some non-equivalence results when the dimensions satisfy $4\leq k\leq \frac{n-4}{2}$ and $\frac{n+4}{2}\leq k\leq n-4$. Finally, we make a conclusion and give some further research problems in Section \ref*{sec:7}.
\section{Preliminaries }
\label{sec:2}
\subsection{Elliptic Function Fields}
Suppose $E$ is a projective, non-singular and geometrically irreducible elliptic curve over the finite field $\mathbb{F}_{q}$. Let $F=\mathbb{F}_{q}(E)$ be the corresponding elliptic function field of $E$. Deote by $E(\mathbb{F}_{q})$ the set of $\mathbb{F}_{q}$-rational points of $E$.
	
	A place $P$ of $F$ is the maximal ideal of some valuation ring $\mathcal{O}_{P}$ of $F/\mathbb{F}_{q}$ and we have the discrete valuation $\mathrm{v}_{P}(\cdot)$. Denote by $\mathbb{P}_{F}=\{P\ |\ P\ \text{is a place of}\ F/\mathbb{F}_{q}\}$ the set of all places on elliptic function field $F/\mathbb{F}_{q}$ and the degree of a place is given by $\deg(P)=[F_{P} :\mathbb{F}_{q}]$, where $F_{P}=\mathcal{O}_{P}/P$.  For any nonzero rational function $f\in F$, divisor $(f)$ is called a principal divisor corresponding to $f$. It has the decomposition $(f)=(f)_{0}-(f)_{\infty}$ where $(f)_{0}=\sum_{P\in\mathbb{P}_{F},\mathrm{v}_{P}(f)>0}\mathrm{v}_{P}(f)P$ and $ (f)_{\infty}=\sum_{P\in\mathbb{P}_{F},\mathrm{v}_{P}(f)<0}(-\mathrm{v}_{P}(f))P   $.
	
	A divisor is a formal sum of places. Denote by $\mathbb{D}_{F}$ the set of all divisors. 
	For divisor $G=\sum_{P\in\mathbb{P}_{F}}\mathrm{v}_{P}(G)P\in \mathbb{D}_{F}$, the support of $G$ is denoted by $\mathrm{Supp}(G)=\{P\in\mathbb{P}_{F}\ |\ \mathrm{v}_{P}(G)\neq0\}$. The degree of divisor $\deg(G)$ is defined as $\deg(G)=\sum_{P\in\mathbb{P}_{F}}\mathrm{v}_{P}(G)\deg(P)$. A divisor $G=\sum_{P\in\mathbb{P}_{F}}\mathrm{v}_{P}(G)P$ is called effective if $\mathrm{v}_{P}(G)\geq 0$ for all $P\in\mathbb{P}_{F}$. The Riemann-Roch space associated to a non-negative divisor $G$
	\begin{displaymath}
		\mathcal{L}(G):=\{f\in\mathbb{F}_{q}(E)\setminus\{0\}\ |\ (f)+G\geq 0\}\cup\{0\}
	\end{displaymath}
	is a finite $\mathbb{F}_{q}$-dimensional vector space, and denote by $\ell(G)$ the dimension of $\mathcal{L}(G)$.

	Let $\Omega_{F}$ be the set of all Weil differentials of $F/\mathbb{F}_{q}$. Denote by $\Omega_{F}(G)$ the $\mathbb{F}_{q}$-vector space of all Weil differentials $\omega$ with the differential divisor $(\omega)\geq G$ together with the zero differential, i.e.
	\begin{displaymath}
		\Omega_{F}(G):=\{\omega\in\Omega_{F}\ |\ \omega=0\ \text{or}\ (\omega)\geq G\}.
	\end{displaymath}
	
	A divisor $W$ is called a canonical divisor of $F/\mathbb{F}_{q}$ if $W=(\omega)$ for some $\omega\in\Omega_{F}$. We have the following proposition in reference \cite{1}.
	\begin{proposition}
		\label{0}
		\begin{enumerate}
			\item If $W$ is a canonical divisor of the elliptic function field $F$, then $\deg(W)=0$ and $\ell(W)=1$.
			\item For $0\neq x\in F$ and $0\neq \omega\in\Omega_{F}$ we have $(x\omega)=(x)+(\omega)$.
		\end{enumerate}
	\end{proposition}

	The well-known Riemann-Roch theorem for the elliptic function field states
	\begin{displaymath}
		\ell(G)= \deg G+\ell(W-G)
	\end{displaymath}
	where $W$ is a canonical divisor for the elliptic function field $F$.
	
	The defining equation of the elliptic function field over finite fields can be classified as follows.
 \begin{lemma} {(see \cite{1})}
 Let $F /\mathbb{F}_{q}$ be an elliptic function field.
 \begin{enumerate}
 	\item If $p$ is odd, there exist $x, y\in F$ such that $F=\mathbb{F}_{q}(x, y)$ and 
 	\begin{displaymath}
 		y^{2}=f(x),
 	\end{displaymath}
 	with a square-free polynomial $f(x)\in \mathbb{F}_{q}[x]$ of degree $3$.
 	\item If $p$ is even, there exist $x, y\in F$ such that $F=\mathbb{F}_{q}(x, y)$ and 
 	\begin{displaymath}
    \vspace{-0.25 cm}
 		y^{2}+y=f(x)\ \text{with\ $\deg(f)=3$},
        \vspace{-0.25 cm}
        \end{displaymath}
 	or \begin{displaymath}
 		y^{2}+y=x+\frac{1}{ax+b}\ \text{with}\ a,b\in \mathbb{F}_{q}\ \text{and}\ a\neq 0.	
        \end{displaymath}
 \end{enumerate}
 \label{1}
\end{lemma}
\begin{remark}
	\begin{enumerate}
		\item By Lemma \ref{1}, it can be checked that the first defining equation belongs to the Kummer extension while the second and third belong to the Artin-Schreier curves. 
		\item For the simplicity, we denote the three kinds of defining equations as the first type, second type and third type respectively in the rest of the discussions.
		\item Denote by $O$ and $P_{\infty}$ the infinity of the corresponding elliptic curve and projective line $\mathbb{P}_{\mathbb{F}_{q}}$ respectively. The valuations at infinity of $x$ and $y$ are given by $v_{O}(x)=-2$ and $v_{O}(y)=-3$ for the first and second types. For the third type, we have $v_{O}(x)=-2$ and $v_{O}(y)=-1$. For the fact that $2v_{O}(y)=v_{O}(y^{2}+y)=v_{O}(x+\frac{1}{ax+b})=e(O|P_{\infty})v_{P_{\infty}}(x+\frac{1}{ax+b})=2v_{P_{\infty}}(x)=-2$, 
			then we have $v_{O}(y)=-1$ where $e(O| P_{\infty})$ is the ramification index of $O$ over $P_{\infty}$.
            \end{enumerate}
\end{remark}
The following Lemma \ref{2} gives an explicit basis of the Riemann-Roch spaces with regard to the elliptic function fields over the finite fields $\mathbb{F}_{q}$ with characteristic $p\geq 3 $ in reference \cite{2}.

\begin{lemma}
{(see \cite{2})}
Notations as above.  For any integer $k\geq 1$, the Riemann-Roch space $\mathcal{L}(kO)$
	has a basis $\{x^{i}y^{j}\ |\ i\in\mathbb{Z}_{\geq 0},\ j\in\{0, 1\}, 2i+3j\leq k\}$.
\label{2}
\end{lemma}
It can be checked that Lemma \ref{2} also holds for the second type of elliptic curves (see \cite{1}). By Lemma \ref{1} and the Weierstrass gap theorem, there is an ascending chain of Riemann-Roch spaces with respect to the corresponding elliptic function field $F$:
\begin{displaymath}
	\mathbb{F}_{q}=\mathcal{L}(0)=	\mathcal{L}(O)\subsetneq\mathcal{L}(2O)\subsetneq\cdots\subsetneq\mathcal{L}(kO).
\end{displaymath}
By Lemma \ref{2}, we have $x^{i}\in\mathcal{L}(2iO)\setminus\mathcal{L}((2i-1)O)$ for $1\leq i\leq \lfloor\frac{k}{2}\rfloor$ and $x^{j}y\in\mathcal{L}((2j+3)O)\setminus\mathcal{L}((2j+2)O)$ for $0\leq j\leq \lfloor\frac{k-3}{2}\rfloor$ with regard to the valuations at infinity $v_{O}(x)=-2$ and  $v_{O}(y)=-3$ respectively.

By the proof in Lemma \ref{1} in \cite{1} , it can be checked that the third type of elliptic curve is birationally equivalent to the elliptic curve with Weierstrass equation:
\begin{equation}
	y_{1}^{2}+(ax_{1}+b)y_{1}=x_{1}^{3}+\epsilon_{1} x_{1}^{2}+\epsilon_{2} x_{1}+\epsilon_{3}
\end{equation}
for some $\epsilon_{i}\in\mathbb{F}_{q}$ and $i=1, 2, 3$ and the birational map can be given by
\begin{displaymath}
	\begin{cases}
		x_{1}=\xi(x)= u_{1}x+v_{1}\\
		y_{1}=\xi(y)= u_{2}(ax+b)y+v_{2}
	\end{cases}
\end{displaymath}
for some $u_{1}, u_{2}\in\mathbb{F}^{*}_{q}$ and $v_{1}, v_{2}\in\mathbb{F}_{q}$ and the valuations of $x_{1}$ and $y_{1}$ at infinity $O$ can be given by $v_{O}(x_{1})=-2$ and $v_{O}(y_{1})=-3$, respectively.

The basis of the Riemann-Roch space $\mathcal{L}(kO)$ for the third type of elliptic function field can be given by the following
\begin{displaymath}
\mathcal{L}(kO)=\left\lbrace 1, x, \cdots, x^{\lfloor\frac{k}{2}\rfloor}, y(ax+b), xy(ax+b), \cdots, x^{\lfloor\frac{k-3}{2}\rfloor}y(ax+b)\right\rbrace.
\end{displaymath}
	Now we introduce the concept of conorm and details can be checked at \cite{1}.
	\begin{definition}
	Notations as above. For a place $P\in\mathbb{P}_{\mathbb{F}_{q}(x)}$,                  its conorm is defined to be
	\begin{displaymath}
		\mathrm{Con}_{F/\mathbb{F}_{q}(x)}(P):=\sum_{P'\mid P}e(P'|P)\cdot P'
	\end{displaymath}
	where $e(P'|P)$ is the ramification index of $P'$ over $P$ and the sum runs over all places $P'\in\mathbb{P}_{F}$ lying over $P$. The conorm map is extended to a homomorphism from $\mathbb{D}_{\mathbb{F}_{q}(x)}$ to $\mathbb{D}_{F}$ by setting
	\begin{displaymath}
		\mathrm{Con}_{F/\mathbb{F}_{q}(x)}(\sum n_{P}\cdot P):=	\sum n_{P}\cdot \mathrm{Con}_{F/\mathbb{F}_{q}(x)}(P).
	\end{displaymath}
	\end{definition}

\subsection{Elliptic Curve Codes (ECCs)}
Let $P_{1}, P_{2}, \cdots, P_{n}$ be $n$ distinct $\mathbb{F}_{q}$-rational points and $D=P_{1}+P_{2}+\cdots+P_{n}$. For any divisor $G\in \mathbb{D}_{F}$ such that $\mathrm{Supp}(G)\cap \mathrm{Supp}(D)=\emptyset$, the ECC $C_{\mathcal{L}}(D, G)$ is defined to be the image of the evaluation map $ev_{\mathcal{L}}$
\begin{displaymath}
	ev_{\mathcal{L}}:\mathcal{L}(G)\to\mathbb{F}_{q}^{n};\ f\mapsto (f(P_{1}), f(P_{2}), \cdots, f(P_{n})).
\end{displaymath}
That is, $C_{\mathcal{L}}(D, G):= ev_{\mathcal{L}}(\mathcal{L}(G))$.

The parameters of ECC $\mathcal{C}_{\mathcal{L}}(D, G)$ are given by code length $n$, dimension $k=\ell(G)-\ell(G-D)$ and the minimum distance $n-\deg(G)\leq d\leq n+1-\deg(G)$.

 For any $P\in E(\mathbb{F}_{q})$, choose one local uniformizer $t$ for $P$. Then for any differential $\omega$, we can write
$\omega= udt$ with some $u\in F$. By writing the $P-$adic expansion for the $u$, we have $u=\sum_{i=i_{0}}^{\infty}a_{i}t^{i}$ for some $i_{0}\in\mathbb{Z}$ and $a_{i}\in\mathbb{F}_{q}$. The residue map of $\omega$ at point $P$ is defined as 
\begin{displaymath}
	\mathrm{res}_{P}(\omega)=\mathrm{res}_{P, t}(u)=a_{-1}.
\end{displaymath}
The residue ECC $C_{\Omega}(D, G)$ is defined to be the image of residue map $ev_{\Omega}$
\begin{displaymath}
	ev_{\Omega}:\Omega_{F}(G)\to\mathbb{F}_{q}^{n};\ \omega\mapsto\left( \mathrm{res}_{P_{1}}(\omega), \mathrm{res}_{P_{2}}(\omega),\cdots, \mathrm{res}_{P_{n}}(\omega)\right).
\end{displaymath}
That is, $C_{\Omega}(D, G):= ev_{\Omega}(\Omega_{F}(G))$.

The parameters of the residue ECCs $\mathcal{C}_{\Omega}(D, G)$ are given by the code length $n$, dimension $k=i(G-D)-i(G)$ and the minimum distance $ \deg(G)\leq d\leq \deg(G)+1$ where $i(G-D)$ and $i(G)$ are the indices of specialty of $G-D$ and $G$ respectively. By the Serre duality theorem in \cite{1}, we have $i(G-D)=\ell(W-(G-D))$ and $i(G)=\ell(W-G)$.

In this paper, we will focus on the single-point ECCs $\mathcal{C}_{\mathcal{L}}(D, kO)$ and $\mathcal{C}_{\Omega}(D, kO)$ for some non-negative integer $k$. Since the dual of an MDS code is still MDS, the minimum distances
\begin{displaymath}
d\left( \mathcal{C}_{\mathcal{L}}(D, G)\right) =n-k+1\ \text{and}\ d\left( \mathcal{C}_{\Omega}(D, G)\right) =k+1
\end{displaymath}
which means they are both MDS codes, or
 \begin{displaymath}
  d\left( \mathcal{C}_{\mathcal{L}}(D, G)\right)=n-k\ \text{and}\ d\left( \mathcal{C}_{\Omega}(D, G)\right) =k
  \end{displaymath}
which means they are both NMDS codes. 

Based on the results of Lemma \ref{2}, we can directly choose the generator matrix of the ECC $\mathcal{C}_{\mathcal{L}}(D, kO)$ as in Lemma 3.6 in \cite{2} for the first two types of ECCs

\begin{small}
\begin{displaymath}
	\begin{pmatrix}
		1& 1& \cdots& 1\\
		\alpha_{1}&\alpha_{2}&\cdots&\alpha_{n}\\
		\vdots&\vdots&\cdots&\vdots\\
		\alpha^{\lfloor\frac{k}{2}\rfloor}_{1}&	\alpha^{\lfloor\frac{k}{2}\rfloor}_{2}&\cdots&	\alpha^{\lfloor\frac{k}{2}\rfloor}_{n}\\
		\beta_{1}&\beta_{2}&\cdots&\beta_{n}\\
		\beta_{1}\alpha_{1}&\beta_{2}\alpha_{2}&\cdots&\beta_{n}\alpha_{n}\\
		\vdots&\vdots&\cdots&\vdots\\
		\beta_{1}\alpha_{1}^{\lfloor\frac{k-3}{2}\rfloor}&	\beta_{2}\alpha_{2}^{\lfloor\frac{k-3}{2}\rfloor}&\cdots&	\beta_{n}\alpha_{n}^{\lfloor\frac{k-3}{2}\rfloor}
	\end{pmatrix}
\end{displaymath}
\end{small}
and the generator matrix for the third type of ECCs
\begin{small}
\begin{displaymath}
	\begin{pmatrix}
		1& 1& \cdots& 1\\
		\alpha_{1}&\alpha_{2}&\cdots&\alpha_{n}\\
		\vdots&\vdots&\cdots&\vdots\\
		\alpha^{\lfloor\frac{k}{2}\rfloor}_{1}&	\alpha^{\lfloor\frac{k}{2}\rfloor}_{2}&\cdots&	\alpha^{\lfloor\frac{k}{2}\rfloor}_{n}\\
		\beta_{1}(a\alpha_{1}+b)&\beta_{2}(a\alpha_{2}+b)&\cdots&\beta_{n}(a\alpha_{n}+b)\\
		\beta_{1}\alpha_{1}(a\alpha_{1}+b)&\beta_{2}\alpha_{2}(a\alpha_{2}+b)&\cdots&\beta_{n}\alpha_{n}(a\alpha_{n}+b)\\
		\vdots&\vdots&\cdots&\vdots\\
		\beta_{1}\alpha_{1}^{\lfloor\frac{k-3}{2}\rfloor}(a\alpha_{1}+b)&	\beta_{2}\alpha_{2}^{\lfloor\frac{k-3}{2}\rfloor}(a\alpha_{2}+b)&\cdots&	\beta_{n}\alpha_{n}^{\lfloor\frac{k-3}{2}\rfloor}(a\alpha_{n}+b)
	\end{pmatrix}.
    \vspace{-0.25 cm}
\end{displaymath}
\end{small}
\subsection{Schur Squares of Linear Codes}
Two codes $\mathcal{C}_{1}$ and $\mathcal{C}_{2}$ are said to be (monomialy) equivalent if $\mathcal{C}_{2}$ can be obtained from $\mathcal{C}_{1}$ by a permutation of coordinates and a component-wise multiplication with some vector $\mathbf{v}=(v_{1}, v_{2}, \cdots, v_{n})\in\left(\mathbb{F}^{*}_{q} \right)^{n}$. 

\begin{definition}
	For two vector $\bm{a}$ and $\bm{b}$ in $\mathbb{F}^{n}_{q}$, the Schur product $\bm{a}\star \bm{b}$ of $\bm{a}$ and $\bm{b}$ is the component-wise product, i.e.,
	\begin{displaymath}
		\bm{a}\star \bm{b}:=(a_{1}b_{1}, a_{2}b_{2}, \cdots, a_{n}b_{n}).
	\end{displaymath}
	For two linear codes $\mathcal{C}_{1}, \mathcal{C}_{2}\subseteq\mathbb{F}_{q}^{n}$, the Schur product $\mathcal{C}_{1}\star\mathcal{C}_{2}$ of $\mathcal{C}_{1}$ and $\mathcal{C}_{2}$ is the linear subspace of $\mathbb{F}_{q}^{n}$ spanned by all the Schur products $\bm{c}_{1}\star \bm{c}_{2}$ with $\bm{c}_{1}\in\mathcal{C}_{1}$ and $\bm{c}_{2}\in\mathcal{C}_{2}$, i.e.,
	\begin{displaymath}
\mathcal{C}_{1}\star\mathcal{C}_{2}:=\mathrm{span}_{\mathbb{F}_{q}}\left \lbrace \bm{c}_{1}\star \bm{c}_{2}:\bm{c}_{1}\in\mathcal{C}_{1}, \bm{c}_{2}\in\mathcal{C}_{2} \right\rbrace. 
	\end{displaymath}
		If $\mathcal{C}_{1}=\mathcal{C}_{2}=\mathcal{C}$, then we call $\mathcal{C}^{\star 2}:=\mathcal{C}\star\mathcal{C}$ the Schur square of $\mathcal{C}$.
\end{definition}

Schur squares are introduced for cryptography analysis and are also used to distinguish a given code from random ones, such as in \cite{7} \textit{et al}.

The dimension of the Schur square of one code is invariant under the equivalence between the codes. For any $[n, k]_{q}-$linear code $\mathcal{C}$, there is an upper bound for the dimension of the Schur square $\dim{\mathcal{C}^{\star 2}}\leq\min\{n, \frac{k(k-1)}{2}\}$.

The following proposition is a criterion in determining whether a code is equivalent to a Reed-Solomon code in reference \cite{17}.
\begin{proposition}[\cite{17}]\label{3}
\vspace{-0.25 cm}
	 Let $\mathcal{C}\subseteq\mathbb{F}_{q}^{n}$ be a linear code with $\dim(\mathcal{C})\leq\frac{n}{2}$. If $\dim\left(\mathcal{C}^{\star 2} \right)\geq 2\dim(\mathcal{C})$, then the code $\mathcal{C}$ is not equivalent to a Reed-Solomon code.
    \vspace{-0.25 cm}
\end{proposition}
A linear MDS code is called non-RS MDS if it is not linearly equivalent to any RS code. The lemma \ref{3} illustrates that any MDS code satisfying $\dim\left(\mathcal{C}^{\star 2} \right)\geq 2k $ is a non-RS MDS code.

\subsection{Twisted Elliptic Curve Codes (TECCs)}

Now we introduce the definition of twisted elliptic curve codes (TECCs). 
\begin{definition}
	Let $D=\{P_{i}=(\alpha_{i}, \beta_{i})\ |\ i\in \{1, \cdots, n\}\}\subseteq E(\mathbb{F}_{q})\backslash\{ O\}$ be the set of $n$ distinct evaluation points on $E$. For two positive integers $\ell, k$ and $\ell\leq\min\{k, n-k\}< n$, suppose that $\mathbf{t}=(t_{1}, t_{2}, \cdots, t_{\ell})$, $1\leq t_{i}\leq\min\{k-1, n-k-1\}$ are distinct and $\mathbf{h}=(h_{1}, h_{2}, \cdots, h_{\ell})$, where $0\leq h_{i}\leq \lfloor\frac{k}{2}\rfloor$ if $k+t_{i}\equiv 1$ mod $2$ or $0\leq h_{j}\leq \lfloor\frac{k-3}{2}\rfloor$ if $k+t_{j}\equiv 0$ mod $2$. Each $h_{i},h_{j}$ are also distinct and $\bm{\eta}=(\eta_{1}, \eta_{2}, \cdots, \eta_{\ell})\in\mathbb{F}^{\ell}_{q}$. Then the twisted elliptic curve codes (TECCs) are defined to be the image of the evaluation map
    $$ev_{D}: S(\bm{h},\bm{t},\bm{\eta})\rightarrow \mathbb{F}^{n}_{q}; f\mapsto (f(P_{1}),f(P_{2}),\cdots,f(P_{n}))$$ and denoted by $\mathcal{C}(D, kO,\bm{h},\bm{t}, \bm\eta)$, where $S(\bm{h},\bm{t},\bm{\eta})$ is called the defining set and it can be divided into the followings:
    \begin{small}
	\begin{displaymath}
		\begin{aligned}
			S(\bm{h},\bm{t},\bm{\eta})=\left\lbrace \sum_{i=0}^{\lfloor\frac{k}{2}\rfloor}a_{i}x^{i}+\sum_{j=0}^{ \lfloor\frac{k-3}{2}\rfloor}b_{j}x^{j}y+\sum_{\substack{s\in\{1, 2, \cdots, \ell\} \\
					k+t_{s}\equiv 1\ \text{mod}\ 2
			}}a_{h_{s}}\eta_{s}x^{\frac{k-3+t_{s}}{2}}y+\sum_{\substack{s\in\{1, 2, \cdots, \ell\} \\
					k+t_{s}\equiv 0\ \text{mod}\ 2}} b_{h_{s}}\eta_{s}x^{ \frac{k+t_{s}}{2}}\right.\\\left.
			\bigg|\ a_{i},b_{j}\in\mathbb{F}_{q},\ \eta\in\mathbb{F}^{*}_{q}\right\rbrace
		\end{aligned}
	\end{displaymath}
    \label{112}
    \end{small}
    for the first and second types of TECCs and the third type of TECCs can be given by   
    \begin{small}
\begin{displaymath}
		\begin{aligned}
			S(\bm{h},\bm{t},\bm{\eta})=\left\lbrace \sum_{i=0}^{\lfloor\frac{k}{2}\rfloor}a_{i}x^{i}+\sum_{j=0}^{ \lfloor\frac{k-3}{2}\rfloor}b_{j}x^{j}(ax+b)y+\sum_{\substack{s\in\{1, 2, \cdots, \ell\} \\
					k+t_{s}\equiv 1\ \mod\ 2
			}}a_{h_{s}}\eta_{s}x^{\frac{k-3+t_{s}}{2}}(ax+b)y+\sum_{\substack{s\in\{1, 2, \cdots, \ell\} \\
					k+t_{s}\equiv 0\ \mod\ 2}}b_{h_{s}}\eta_{s}x^{ \frac{k+t_{s}}{2}}\right.\\\left.
			\bigg|\ a_{i},b_{j}\in\mathbb{F}_{q},\ \eta\in\mathbb{F}^{*}_{q}\right\rbrace.
		\end{aligned}
	\end{displaymath}
    \end{small}
    \end{definition}
We assume $k\geq 3$ in the rest of the discussions. By calculating the indices of the poles, it can be checked that the rational functions
\begin{small}
\begin{displaymath}
	\begin{cases}
	x^{i}& \text{for}\ i\in\{0, \cdots, \lfloor\frac{k}{2}\rfloor\}/\{h_{s}\ |\ h_{s}\in\{h_{1}, \cdots, h_{\ell}\}\ \text{with}\ k+t_{s}\equiv 1\ \text{mod}\ 2\ \}\\
	 x^{h_{s}}+\eta_{s}x^{\frac{k-3+t_{s}}{2}}y&\text{for}\ s\in\{1, \cdots, \ell\}\ \text{and}\ k+t_{s}\equiv 1\ \mod\ 2\\
	 x^{j}y&\text{for}\ j\in\{0, \cdots, \lfloor\frac{k-3}{2}\rfloor\}/\
\{h_{s}\ |\ h_{s}\in\{h_{1},\cdots,h_{\ell}\}\ \text{with}\ k+t_{s}\equiv 0\ \text{mod}\ 2\ \}\\
	 x^{h_{s}}y+\eta_{s}x^{\frac{k+t_{s}}{2}}&\text{for}\ s\in\{1, \cdots, \ell\}\ \text{and}\ k+t_{s}\equiv 0\ \mod\ 2
	 \end{cases}
	\end{displaymath}
    \end{small}
    and 
    \begin{small}
		\begin{displaymath}
		\begin{cases}
			x^{i}& \text{for}\ i\in\{0, \cdots, \lfloor\frac{k}{2}\rfloor\}/\{h_{s}\ |\ h_{s}\in\{h_{1}, \cdots, h_{\ell}\}\ \text{with}\ k+t_{s}\equiv 1\ \text{mod}\ 2\ \}\\
            x^{h_{s}}+\eta_{s}x^{\frac{k-3+t_{s}}{2}}(ax+b)y&\text{for}\ s\in\{1, \cdots, \ell\}\ \text{and}\ k+t_{s}\equiv 1\ \mod\ 2\\
			x^{j}y&\text{for}\ j\in\{0, \cdots, \lfloor\frac{k-3}{2}\rfloor\}/\{h_{s}\ |\ h_{s}\in\{h_{1}, \cdots, h_{\ell}\}\ \text{with}\ k+t_{s}\equiv 0\ \text{mod}\ 2\ \}\\
            x^{h_{s}}(ax+b)y+\eta_{s}x^{\frac{k+t_{s}}{2}}&\text{for}\ s\in\{1, \cdots, \ell\}\ \text{and}\ k+t_{s}\equiv 0\ \mod\ 2
		\end{cases}
	\end{displaymath}
    \end{small}
    are linearly independent over $\mathbb{F}_{q}$ and they form a basis for the two cases of $S(\bm{h},\bm{t},\bm{\eta})$.

In this paper, we shall focus on the following specific TECCs with one twist.  The defining set is divided into the following two cases with respect to the parity of $k$:
\begin{enumerate}
\item For odd $k$ and $0\leq\ell\leq\frac{k-3}{2}$:
\begin{small}
\begin{displaymath}
	S_{\ell}^{(1)}=\left\lbrace \sum_{i=0}^{\frac{k-1}{2}}a_{i}x^{i}+\sum_{j=0}^{ \frac{k-3}{2}}b_{j}x^{j}y+b_{\ell}\eta x^{\frac{k+1}{2}}\ \bigg|\ a_{i},b_{j}\in\mathbb{F}_{q},\ \eta\in\mathbb{F}^{*}_{q}\right\rbrace. 
\end{displaymath}
\end{small}

\item For even $k$ and $0\leq\ell\leq\frac{k}{2}$:
\begin{small}
\begin{displaymath}
	S_{\ell}^{(2)}=\left\lbrace \sum_{i=0}^{ \frac{k}{2}}a_{i}x^{i}+\sum_{j=0}^{ \frac{k-4}{2}}b_{j}x^{j}y+a_{\ell}\eta x^{\frac{k-2}{2}}y\ \bigg|\ a_{i},b_{j}\in\mathbb{F}_{q},\ \eta\in\mathbb{F}^{*}_{q}\right\rbrace. 
\end{displaymath}
\end{small}
\end{enumerate}
As for the TECCs constructed over the third type of elliptic curves, the defining sets are divided:
\begin{enumerate}
\item For odd $k$ and $0\leq\ell\leq\frac{k-3}{2}$:
\begin{small}
	\begin{displaymath}
		S_{\ell}^{(1)}=\left\lbrace \sum_{i=0}^{\frac{k-1}{2}}a_{i}x^{i}+\sum_{j=0}^{ \frac{k-3}{2}}b_{j}x^{j}(ax+b)y+ b_{\ell}\eta x^{\frac{k+1}{2}}\ \bigg|\ a_{i},b_{j}\in\mathbb{F}_{q},\ \eta\in\mathbb{F}^{*}_{q}\right\rbrace. 
	\end{displaymath}
    \end{small}
\item For even $k$ and $0\leq\ell\leq\frac{k}{2}$:
    \begin{small}
	\begin{displaymath}
		S_{\ell}^{(2)}=\left\lbrace \sum_{i=0}^{ \frac{k}{2}}a_{i}x^{i}+\sum_{j=0}^{ \frac{k-4}{2}}b_{j}x^{j}(ax+b)y+a_{\ell}\eta x^{\frac{k-2}{2}}(ax+b)y\ \bigg|\ a_{i},b_{j}\in\mathbb{F}_{q},\ \eta\in\mathbb{F}^{*}_{q}\right\rbrace. 
	\end{displaymath}
    \end{small}
    \end{enumerate}

    In the following discussions, we consider the TECCs with one twist:
\begin{displaymath}
	\mathcal{C}^{(i)}(D, kO,\ell, \eta)=\{ev_{D}(f(x)):\ f(x)\in S_{\ell}^{(i)} \}\ \text{ for $i=1,2$.}
\end{displaymath}

\section{The Explicit Construction for the Parity-Check Matrices of ECCs}
\label{sec:4}

In this section, we will first give an explicit construction for the parity-check matrices of ECCs which is important in computing the dual codes of TECCs. 

By Lemma \ref{1}, we know that the elliptic function fields belong to Kummer extensions and Artin-Schreier extensions. The following proposition calculates the differents between the elliptic function fields and rational function fields.

\begin{proposition} 
(see \cite{1}) Notations as above. 
\begin{enumerate}
\item  $p$ is odd. Suppose that $F=\mathbb{F}_{q}(x, y)$ with
\begin{displaymath}
	y^{2}=f(x)\in \mathbb{F}_{q}[x],
\end{displaymath}
where $f(x)$ is a square-free polynomial of degree $3$. Consider the decomposition $f(x)=c\prod_{i=1}^{r}p_{i}(x)$ of $f(x)$ into irreducible monic polynomials $p_{i}(x)\in \mathbb{F}_{q}[x]$ with $0\neq c\in \mathbb{F}_{q}$. Denote by $Q_{i}\in \mathbb{P}_{\mathbb{F}_{q}(x)}$ the place of $\mathbb{F}_{q}(x)$ corresponding to $p_{i}(x)$ and $Q_{\infty}$ the pole of $x$. Then the following hold:
\begin{enumerate}
	\item  $\mathbb{F}_{q}$ is the full constant field of $F$ and $F/\mathbb{F}_{q}(x)$ is an elliptic function field.
	\item  The extension $F/\mathbb{F}_{q}(x)$ is cyclic of degree $2$. The places $Q_{1},\cdots, Q_{r}$ and $Q_{\infty}$ are ramified in $F/\mathbb{F}_{q}(x)$; each of them has exactly one extension in $F$, say $S_{1}, \cdots, S_{r}$ and $S_{\infty}$, and we have the ramification indices of the places $e(S_{i}|Q_{j})=e(O|Q_{\infty})=2$, $\deg S_{j}=\deg Q_{j}$ and $\deg O=1$.
	\item  $Q_{1},\cdots, Q_{r}$ and $Q_{\infty}$ are the only places of $\mathbb{F}_{q}(x)$ which are ramified in $F/\mathbb{F}_{q}(x)$, and the different of $F/\mathbb{F}_{q}(x)$ is 
	\begin{displaymath}
		\mathrm{Diff}(F/\mathbb{F}_{q}(x))=S_{1}+\cdots+S_{r}+O.
	\end{displaymath}
\end{enumerate}

\item  $p$ is even. Suppose that $F=\mathbb{F}_{q}(x, y)$ with
	\begin{displaymath}
	y^{2}+y=f(x)\ \text{and}\deg f=3,
\end{displaymath}
or
\begin{displaymath}
	y^{2}+y=x+\frac{1}{ax+b}\ \text{for some}\ a,b\in \mathbb{F}_{q}\ \text{and}\ a\neq 0.
\end{displaymath}
 Denote by $Q_{\infty}\in \mathbb{P}_{\mathbb{F}_{q}(x)}$ the pole of $x$ in $\mathbb{F}_{q}(x)$ and by $Q'\in \mathbb{P}_{\mathbb{F}_{q}(x)}$ the place corresponding to the polynomial $ax+b$ in $\mathbb{F}_{q}(x)$. Then the following hold:
\begin{enumerate}
	\item $\mathbb{F}_{q}$ is the full constant field of $F$ and $F/\mathbb{F}_{q}(x)$ is an elliptic function field.
	\item The extension $F/\mathbb{F}_{q}(x)$ is cyclic of degree $2$. The only places of $\mathbb{F}_{q} (x)$ which ramify in $F/\mathbb{F}_{q}(x)$ are
	\begin{displaymath}
		\begin{cases}
			Q_{\infty},& \text{in\ case\ (2)}\\
			Q_{\infty}\ \text{and}\ Q',& \text{in\ case\ (3)}.
		\end{cases}
	\end{displaymath}
	Let $O$ (resp. $S'$) be the place of $F$ lying over $Q_{\infty}$ (resp. $Q'$). Then $\deg O=\deg S'=1$ and 
	\begin{displaymath}
			\mathrm{Diff}(F/\mathbb{F}_{q}(x))=
		\begin{cases}
	4O,&\text{in\ case\ (2)}\\
	2O+2S',&\text{in\ case\ (3)}.
		\end{cases}
	\end{displaymath}
\end{enumerate}
\end{enumerate}
\end{proposition}

\begin{theorem}
\cite{15}
Let $\omega\in\Omega_{F}$ be a canonical differential such that $v_{P_{i}}(\omega)=-1$ for $i=1, 2, \cdots, n$. Let $H:=D-G+(\omega)$ and $\bm\gamma:=(\mathrm{res}_{P_{1}}(\omega), \cdots, \mathrm{res}_{P_{n}}(\omega))$. Then 
\begin{displaymath}
	\mathcal{C}_{\mathcal{L}}(D, G)^{\perp}=\bm\gamma\star\mathcal{C}_{\mathcal{L}}(D, H).
    \label{11}
    \vspace{-0.25 cm}
\end{displaymath}

\end{theorem}
Note that such a canonical differential in the above theorem exists (see \cite{15}).

	The structure of rational points over the elliptic curves is important for the constructions of ECCs. To overcome the limits of the code lengths of GRS codes, we need to construct ECCs over the elliptic curves with sufficient rational points. The well-known Hasse-Weil bound for the numbers of the rational points $\#E(\mathbb{F}_{q})$ over the elliptic curve $E$ states $\left| \# E(\mathbb{F}_{q})-\left( 1+q\right) \right| \leq 2\sqrt{q}$. 
    
    An algebraic curve is called maximal if it's number of rational points attains Hasse-Weil upper bound $q+1+2\sqrt{q}$. 
    In this case, the Zeta function 
	is completely determined by
   	\begin{displaymath}
		Z(t)=\frac{1+at+qt^{2}}{(1-t)(1-qt)}=\frac{(1+\sqrt{q}t)^{2}}{(1-t)(1-qt)}
	\end{displaymath}      
where $a=\#E(\mathbb{F}_{q})-(q+1)$ and $t=q^{-s}$ for some integer $s$ .

Denote by $\mathcal{O}_{E}$ the algebraic integral ring of elliptic function field $E$. Then any place $P\in\mathbb{P}_{\mathbb{F}_{q}(x)}$ have the following decomposition in $E$.
\begin{enumerate}
    \item Splitting completely. $P\mathcal{O}_{E}=P_{1}P_{2}$ for some $P_{1}\neq P_{2}\in\mathbb{P}_{E}$.
    \item Ramification. $P\mathcal{O}_{E}=P^{2}_{1}$ for some $P_{1}\in\mathbb{P}_{E}$.
    \item Staying inertia. $P\mathcal{O}_{E}=P_{1}$ for some $P_{1}\in\mathbb{P}_{E}$.
    \end{enumerate}

Denote by $T$ the set of rational places in rational function field $\mathbb{F}_{q}(x)$ such that the any place $P\in T$ splits completely or being inertia in the elliptic function field $F$ and we take $D=\mathrm{Con}_{F/\mathbb{F}_{q}(x)}(\sum_{P\in T} P)$. For the simplicity of the discussions and the study of the self-duality, we let $\#\mathrm{Supp}(D)$ be even for the constructions of the TECCs.

\begin{remark}
    If we take $q\geq 4^{n}\times n^{2}$ and $q\equiv 1\ (\text{mod}\ 4)$, then there exists $n$ distinct elements $\alpha_{1},\cdots,\alpha_{n}$ such that $\alpha_{i}-\alpha_{j}$ are nonzero squares in $\mathbb{F}_{q}$ for all $1\leq i<j\leq n$ and each $\alpha_{i}$ gives two points $(\alpha_{i},\beta_{i})$ and $(\alpha_{i},-\beta_{i})$ for some $\beta_{i}\in\mathbb{F}_{q}$, see \cite{10}. Then we have the divisor $D$ with the pre-image splitting completely in the elliptic function fields.
\end{remark}

Denote by $x(E(\mathbb{F}_{q}))$ the $x-$component of rational points on elliptic curve $E$ except the point at infinity $O$. When $p$ is odd, we have the classifications:
\begin{enumerate}
	\item $r=1$: We consider the subset $T\subseteq x(E(\mathbb{F}_{q}));$
	\item $r=2$: We denote $\rho_{1}$ as the zero of the polynomials $p_{1}(x)$ and $\deg(p_{1}(x))=1$ in $\mathbb{F}_{q}(x)$ and consider the subset $T\subseteq x(E(\mathbb{F}_{q}))\setminus\{\rho_{1}\};$
	\item $r=3$: We denote $\rho_{i}$ as the zero of the polynomials $p_{i}(x)$ and $\deg(p_{i}(x))=1$ in $\mathbb{F}_{q}(x)$ for $i=1,2 ,3$ and consider the subset $T\subseteq x(E(\mathbb{F}_{q}))\setminus \{\rho_{1}, \rho_{2}, \rho_{3}\}$. 
\end{enumerate}

When $p$ is even, for the first case, we consider arbitrary a subset $T\subseteq x(E(\mathbb{F}_{q}))$ and we denote by $\rho'$ the place corresponding to the polynomial $ax+b$ in $\mathbb{F}_{q}(x)$ for the second case, then we have $T\subseteq x(E(\mathbb{F}_{q}))\setminus\{\rho'\}$.

Consider the local uniformizer $t=\prod_{\alpha\in T }(x-\alpha)$ satisfying $v_{P_{i}}(t)=1$  for $i=1, 2,\cdots, n$.  By the Riemann-Hurwitz formula (see \cite{1}), the differential divisor $(dx)$ can be given by 
\begin{displaymath}
	(dx)=-2(x)_{\infty}+\text{Diff}(F/\mathbb{F}_{q}(x)).
\end{displaymath}
Notice that $(x)_{\infty}=2O$ and $(y)_{\infty}=3O$ for the first and second types and $(x)_{\infty}=2O$ and $(y)_{\infty}=O$ for the third type. Then the differential divisor $(dx)$ can be classified in the following:
\begin{enumerate}
	\item The first type: $(dx)=-2(x)_{\infty}+\text{Diff}(F/\mathbb{F}_{q}(x))=S_{1}+\cdots+S_{r}-3O=(y);$
	\item The second type: $(dx)=-2(x)_{\infty}+\text{Diff}(F/\mathbb{F}_{q}(x))=0;$
	\item The third type: $(dx)=-2(x)_{\infty}+\text{Diff}(F/\mathbb{F}_{q}(x))=2S'-2O=(ax+b)$.
\end{enumerate}

Based on the results above, the canonical divisor $(\omega)=(dx/t)$ can be given by the following
\begin{enumerate}
	\item The first type: the differential divisor $(dx/t)=(dx)-(t)=-D+nO+(y)$ and the divisor $D-G+(dx/t)=(y)+(n-k)O;$
	\item The second type: the differential divisor $(dx/t)=(dx)-(t)=-D+nO$ and the divisor $D-G+(dx/t)=(n-k)O;$
	\item The third type: the differential divisor $(dx/t)=(dx)-(t)=-D+nO+(ax+b)$ and the divisor $D-kO+(dx/t)=(ax+b)+(n-k)O$.
\end{enumerate}

\begin{remark}
    For the second type of elliptic function fields, we have $(dx/t)=0$. Then the differential $\omega=dx/t$ is a canonical divisor defined in Proposition \ref{0}. Then the residue at all rational valuation places $P\in D$ are equal and we can assume $\mathrm{res}_{P}(\omega)=c\in\mathbb{F}_{q}$. By the residue theorem on the algebraic curves in \cite{1}, we have $n\cdot c\equiv 0$ (mod $2$). For code length satisfying $n=2k$, we have $\mathrm{res}_{P}(\omega)=1$ for each $P\in D$ and the following Corollary holds immediately.
\end{remark}

\begin{corollary}\label{Cor:2ndSelfdual}
The ECCs constructed over the second type of elliptic curves are definitely self-dual for $n=2k$.
\end{corollary}

The three types of parity-check matrices of ECC $\mathcal{C}_{\mathcal{L}}(D, kO)$ shall be deduced from following lemma.
\begin{lemma}
	\label{11}
Notations as above. For any integer $k\geq 1$,
	\begin{enumerate}
		\item The basis of Riemann-Roch space $\mathcal{L}(D-G+(dx/t))=\mathcal{L}((y)+(n-k)O)$ over the first type of elliptic function fields can be given by
		\begin{displaymath}
			\begin{split}
				\bigg\{1/y, x/y, \cdots,  x^{\lfloor\frac{n-k}{2}\rfloor}/y, 1, \cdots, x^{\lfloor\frac{n-k-3}{2}\rfloor}\bigg\}.
			\end{split}
		\end{displaymath}
		\item The basis of Riemann-Roch space $\mathcal{L}(D-G+(dx/t))=\mathcal{L}((n-k)O)$ over the second type of elliptic function fields can be given by 
		\begin{displaymath}
			\bigg\{1, x, \cdots,  x^{\lfloor\frac{n-k}{2}\rfloor}, y, \cdots, x^{\lfloor\frac{n-k-3}{2}\rfloor}y\bigg\}.
		\end{displaymath}
		\item The basis of Riemann-Roch space $\mathcal{L}(D-G+(dx/t))=\mathcal{L}((ax+b)+(n-k)O)$ over the third type of elliptic function fields can be given by
		\begin{displaymath}
			\bigg\{1/(ax+b), x/(ax+b), \cdots, x^{\lfloor\frac{n-k}{2}\rfloor}/(ax+b), y, xy, \cdots, x^{\lfloor\frac{n-k-3}{2}\rfloor}y\bigg\}.
            \end{displaymath}
	\end{enumerate}
\end{lemma}
\begin{proof}
\begin{enumerate}
\item For the first type of elliptic function fields, the divisor $(y)$ can be classified by
\begin{displaymath}
	(y)=
	\begin{cases}
		S_{1}-3O,\ &\text{for}\ r=1\\
		S_{1}+S_{2}-3O,\ & \text{for}\ r=2\\
		S_{1}+S_{2}+S_{3}-3O,\ &\text{for}\ r=3.\\
	\end{cases}
\end{displaymath}
Substitute the three cases $S_{1}$, $S_{1}+S_{2}$ and $S_{1}+S_{2}+S_{3}$ as $(y)+3O$, then we have Riemann-Roch space $\mathcal{L}(D-G+(dx/t))=\mathcal{L}((y)+(n-k)O)$. 

For any $0\neq f\in\mathcal{L}((y)+(n-k)O)$, we consider the mapping $\psi:\ f\mapsto fy$ and it can be easily checked that $\psi$ is bijective, then there is an induced mapping between two Riemann-Roch spaces $\psi:\ \mathcal{L}(D-G+(dx/t))=\mathcal{L}((y)+(n-k)O)\to \mathcal{L}((n-k)O)$. By Lemma \ref{1}, the basis of $\mathcal{L}((y)+(n-k)O)$ can be given by
$\{1/y, x/y, \cdots,  x^{\lfloor\frac{n-k}{2}\rfloor}/y, 1, \cdots, x^{\lfloor\frac{n-k-3}{2}\rfloor}\}.
$

\item For the second type of elliptic function fields, we have that $\mathcal{L}(D-G+(dx/t))=\mathcal{L}((n-k)O)$ and the basis can be given by Lemma \ref{1} directly.

\item For the third type of elliptic function fields, we have $(ax+b)=2S'-2O$ and $2S'+(n-k-2)O=(ax+b)+(n-k)O$. For any $f\in\mathcal{L}((ax+b)+(n-k)O)$, we consider the mapping $\phi:f\mapsto (ax+b)f$ and it can be checked $\phi$ is bijective, then there is an induced mapping between two Riemann-Roch spaces $\phi:\mathcal{L}(D-G+(dx/t))=\mathcal{L}((ax+b)+(n-k)O)\to\mathcal{L}((n-k)O)$. Then, the basis of the Riemann-Roch space $\mathcal{L}((ax+b)+(n-k)O)$ can be given by 
$\{1/(ax+b), x/(ax+b), \cdots, x^{\lfloor\frac{n-k}{2}\rfloor}/(ax+b), y, xy, \cdots, x^{\lfloor\frac{n-k-3}{2}\rfloor}y\}.$
\end{enumerate}
\end{proof}

\begin{theorem}[Orthogonality Relations]\label{cc}
		Notations as above. Denote by $\mathrm{res}_{P_{i}}(\omega)=\gamma_{i}$ the residue of differential $\omega$ at valuation place $P_{i}$ for $i=1,\cdots, n$. Then we have the following orthogonality relations for ECCs:
		\begin{enumerate}
			\item The first type: \begin{enumerate}
				\item $\sum_{j=1}^{n}\frac{\gamma_{j}\alpha^{i}_{j}}{\beta_{j}}=0$ for $0\leq i\leq \lfloor\frac{k}{2}\rfloor+\lfloor\frac{n-k}{2}\rfloor;$
				\item $\sum_{j=1}^{n}\gamma_{j}\alpha^{i}_{j}=0$ for $0\leq i\leq\frac{n}{2}-2; $
				\item $\sum_{j=1}^{n}\gamma_{j}\alpha^{i}_{j}\beta_{j}=0$ for $0\leq i\leq\lfloor\frac{k-3}{2}\rfloor+\lfloor\frac{n-k-3}{2}\rfloor$.
			\end{enumerate}
			\item The second type: 
			\begin{enumerate}
				\item $\sum_{j=1}^{n}\gamma_{j}\alpha^{i}_{j}=0$ for $0\leq i\leq \lfloor\frac{k}{2}\rfloor+\lfloor\frac{n-k}{2}\rfloor;$
				\item $\sum_{j=1}^{n}\gamma_{j}\alpha^{i}_{j}\beta_{j}=0$ for $0\leq i\leq\frac{n}{2}-2; $
				\item $\sum_{j=1}^{n}\gamma_{j}\alpha^{i}_{j}\beta^{2}_{j}=0$ for $0\leq i\leq \lfloor\frac{k-3}{2}\rfloor+\lfloor\frac{n-k-3}{2}\rfloor$.
			\end{enumerate}
			\item The third type:
			\begin{enumerate}
				\item $\sum_{j=1}^{n}\frac{\gamma_{j}\alpha^{i}_{j}}{a\alpha_{j}+b}=0$ for $0\leq i\leq \lfloor\frac{k}{2}\rfloor+\lfloor\frac{n-k}{2}\rfloor;$
				\item $\sum_{j=1}^{n}\gamma_{j}\alpha^{i}_{j}\beta_{j}=0$ for $0\leq i\leq \frac{n}{2}-2;$
				\item $\sum_{j=1}^{n}\gamma_{j}\alpha^{i}_{j}(a\alpha_{i}+b)\beta^{2}_{j}=0$ for $0\leq i\leq \lfloor\frac{k-3}{2}\rfloor+\lfloor\frac{n-k-3}{2}\rfloor$.
			\end{enumerate}
		\end{enumerate}
		\begin{proof} Choose one differential divisor $(\omega)$ and $g\in\mathcal{L}(kO)$. 
		By the Serre duality theorem in \cite{1}, we have the isomorphism
		\begin{displaymath}
			\mu: \mathcal{L}((\omega)-(kO-D))\to\Omega_{F}(kO-D);\ f\mapsto f\omega.
		\end{displaymath}
		By the residue theorem in \cite{1}, we have 
        \begin{displaymath}
        \sum_{P\in\mathbb{P}_{F}}fg\omega_{ P}=\sum_{i=1}^{n}f(P_{i})g(P_{i})\mathrm{res}_{P_{i}}(\omega)=0.
\end{displaymath}
		By the results in Lemma \ref{11}, then we have the orthogonality relations.
        \end{proof}
	\end{theorem}
Based on the Lemma \ref{11} and Theorem \ref{cc}, we have the following Corollary.

\begin{corollary}\label{c1}
Notation as above. The parity check matrices of the ECC $\mathcal{C}_{\mathcal{L}}(D, kO)$ can be given in the following three cases.
\begin{itemize}
	\item The first type:
\begin{small}
	\begin{displaymath}
      G(D,\omega, (y)+(n-k)O)=		\begin{pmatrix}
			\frac{	\gamma_{1}}{\beta_{1}}&	\frac{\gamma_{2}}{\beta_{2}}&\cdots&		\frac{\gamma_{n}}{\beta_{n}}\\
			\frac{	\gamma_{1}\alpha_{1}}{\beta_{1}}&	\frac{\gamma_{2}\alpha_{2}}{\beta_{2}}&\cdots&		\frac{\gamma_{n}\alpha_{n}}{\beta_{n}}\\
			\vdots&\vdots&\cdots&\vdots\\
			\frac{\gamma_{1}\alpha^{\lfloor\frac{n-k}{2}\rfloor}_{1}}{\beta_{1}}&		\frac{\gamma_{2}\alpha^{\lfloor\frac{n-k}{2}\rfloor}_{2}}{\beta_{2}}&\cdots&		\frac{\gamma_{n}\alpha^{\lfloor\frac{n-k}{2}\rfloor}_{n}}{\beta_{n}}\\
			\gamma_{1}&	\gamma_{2}&\cdots&		\gamma_{n}\\
			\gamma_{1}\alpha_{1}&\gamma_{2}\alpha_{2}&\cdots&\gamma_{n}\alpha_{n}\\
			\vdots&\vdots&\cdots&\vdots\\
			\gamma_{1}\alpha^{\lfloor\frac{n-k-3}{2}\rfloor}_{1}&		\gamma_{2}\alpha^{\lfloor\frac{n-k-3}{2}\rfloor}_{2}&\cdots&		\gamma_{n}\alpha^{\lfloor\frac{n-k-3}{2}\rfloor}_{n}\\
		\end{pmatrix}.
	\end{displaymath}
    \end{small}
	\item The second type:
    \begin{small}
	\begin{displaymath}
     G(D,\omega, (n-k)O)=
     \begin{pmatrix}
			\gamma_{1}& \gamma_{2}& \cdots& \gamma_{n}\\
			\gamma_{1}\alpha_{1}&\gamma_{2}\alpha_{2}&\cdots&\gamma_{n}\alpha_{n}\\
			\vdots&\vdots&\cdots&\vdots\\
			\gamma_{1}\alpha^{\lfloor\frac{n-k}{2}\rfloor}_{1}&	\gamma_{2}\alpha^{\lfloor\frac{n-k}{2}\rfloor}_{2}&\cdots&	\gamma_{n}\alpha^{\lfloor\frac{n-k}{2}\rfloor}_{n}\\
			\gamma_{1}\beta_{1}&\gamma_{2}\beta_{2}&\cdots&\gamma_{n}\beta_{n}\\
			\gamma_{1}\alpha_{1}\beta_{1}&\gamma_{2}\alpha_{2}\beta_{2}&\cdots&\gamma_{n}\alpha_{n}\beta_{n}\\
			\vdots&\vdots&\cdots&\vdots\\
			\gamma_{1}\alpha_{1}^{\lfloor\frac{n-k-3}{2}\rfloor}\beta_{1}&	\gamma_{2}\alpha_{2}^{\lfloor\frac{n-k-3}{2}\rfloor}\beta_{2}&\cdots&	\gamma_{n}\alpha_{n}^{\lfloor\frac{n-k-3}{2}\rfloor}\beta_{n}
		\end{pmatrix}.
		\end{displaymath}
        \end{small}
	\item The third type:
    \begin{small}
		\begin{displaymath}
          G(D,\omega, (ax+b)+(n-k)O)=		\begin{pmatrix}
			\frac{\gamma_{1}}{a\alpha_{1}+b}& \frac{\gamma_{2}}{a\alpha_{2}+b}& \cdots& \frac{\gamma_{n}}{a\alpha_{n}+b}\\
			\frac{\gamma_{1}\alpha_{1}}{a\alpha_{1}+b}& \frac{\gamma_{2}\alpha_{2}}{a\alpha_{2}+b}& \cdots& \frac{\gamma_{n}\alpha_{n}}{a\alpha_{n}+b}\\
			\vdots&\vdots&\cdots&\vdots\\
			\frac{\gamma_{1}\alpha^{\lfloor\frac{n-k}{2}\rfloor}_{1}}{a\alpha_{1}+b}& 	\frac{\gamma_{1}\alpha^{\lfloor\frac{n-k}{2}\rfloor}_{2}}{a\alpha_{2}+b}& \cdots& 	\frac{\gamma_{1}\alpha^{\lfloor\frac{n-k}{2}\rfloor}_{n}}{a\alpha_{n}+b}\\
		\gamma_{1}\beta_{1}&\gamma_{2}\beta_{2}&\cdots&	\gamma_{n}\beta_{n}\\
		\gamma_{1}\beta_{1}\alpha_{1}&	\gamma_{2}\beta_{2}\alpha_{2}&\cdots&	\gamma_{n}\beta_{n}\alpha_{n}\\
			\vdots&\vdots&\cdots&\vdots\\
			\gamma_{1}\beta_{1}\alpha_{1}^{\lfloor\frac{n-k-3}{2}\rfloor}&		\gamma_{2}\beta_{2}\alpha_{2}^{\lfloor\frac{n-k-3}{2}\rfloor}&\cdots&\gamma_{n}\beta_{n}\alpha_{n}^{\lfloor\frac{n-k-3}{2}\rfloor}
		\end{pmatrix}.
	\end{displaymath}
    \end{small}
\end{itemize}
\end{corollary}

\section{The Dual Codes of TECCs}
\label{sec:5}
In this section, we first give the parity-check matrices for the TECCs and then determine the self-dual conditions for the TECCs.  

In the following, we take the conorm of the rational divisors in $\mathbb{F}_{q}(x)$ which split or stay inertia in the elliptic function fields. In particular, we take even $n$ and denote by $n=2s$ for some positive integer $s\in\mathbb{Z}_{>0}$ for better discussions of self-duality.

For any two codewords $\mathbf{a}=(a_{1}, a_{2}, \cdots, a_{n}), \mathbf{b}=(b_{1}, b_{2}, \cdots, b_{n})\in\mathbb{F}_{q}^{n}$, we consider the Euclidean inner product of $\mathbf{a}$ and $\mathbf{b}$. Then the dual code of a linear code $\mathcal{C}$ is defined to be 
\begin{displaymath}
	\mathcal{C}^{\perp}=\left\lbrace \mathbf{a}\in\mathbb{F}_{q}^{n} :\ \mathbf{a}\cdot \mathbf{c}=\sum_{i=1}^{n}c_{i}a_{i}=0\ \text{for\ all}\ \mathbf{c}\in\mathcal{C}\right\rbrace.
\end{displaymath}

\begin{theorem} 
Notations as above. The parity check matrix of the TECC $\mathcal{C}(D, kO,\ell, \eta)$ constructed over the first type of elliptic curves can be classified as follows.
\begin{enumerate}
	\item 
	The parity check matrix of $\mathcal{C}(D, kO,\ell, \eta)$ with odd $k$: 
	
	Let $a_{0}^{(1)}, a_{1}^{(1)},\cdots, a^{(1)}_{\frac{k-3}{2}-\ell}, b^{(1)}_{0}, b^{(1)}_{1}\cdots, b^{(1)}_{\frac{k-3}{2}-\ell}\in\mathbb{F}_{q}$ be defined by the following recursion
    \begin{small}
	\begin{displaymath}
		\begin{cases}
			a^{(1)}_{r}=-\frac{\sum_{m=1}^{n}\frac{\gamma_{m}}{\beta_{m}}\sum_{i=0}^{r-1}a_{i}^{(1)}\alpha^{s-\ell-1+\frac{k-1}{2}-i}_{m}+\sum_{m=1}^{n}\gamma_{m}\sum_{j=0}^{r-1}b_{j}^{(1)}\alpha^{s-\ell-3+\frac{k-1}{2}-j}_{m}}{\sum_{m=1}^{n}\frac{\gamma_{m}}{\beta_{m}}\alpha_{m}^{s-\ell-1+\frac{k-1}{2}-r}}\\
			b^{(1)}_{r-1}=-\frac{\sum_{m=1}^{n}\gamma_{m}\sum_{i=0}^{r}a_{i}^{(1)}\alpha^{s-\ell-1+\frac{k-3}{2}-i}_{m}+\sum_{m=1}^{n}\gamma_{m}\sum_{j=0}^{r-2}b_{j}^{(1)}\alpha^{s-\ell-3+\frac{k-3}{2}-j}_{m}\beta_{m}}{\sum_{m=1}^{n}\gamma_{m}\alpha_{m}^{s-\ell-3+\frac{k-3}{2}-r}\beta_{m}}\\
        	\end{cases}
	\end{displaymath}
    \end{small}
   for $r=1, 2, \cdots, \frac{k-3}{2}-\ell$ with
   \begin{small}
		\begin{displaymath}
		\begin{cases}
			a^{(1)}_{0}=1\\
			b^{(1)}_{\frac{k-3}{2}-\ell}=-\frac{\sum_{m=1}^{n}\gamma_{m}\alpha^{s-1}_{m}+\eta\sum_{m=1}^{n}\frac{\gamma_{m}}{\beta_{m}}\sum^{\frac{k-3}{2}-\ell}_{i=0}a^{(1)}_{i}\alpha^{s-\ell-1+\frac{k+1}{2}-i}_{m}+\eta\sum_{m=1}^{n}\gamma_{m}\sum^{\frac{k-3}{2}-\ell-1}_{j=0}b_{j}^{(1)}\alpha^{s-\ell-3+\frac{k+1}{2}-j}_{m}}{\eta\sum_{m=1}^{n}\gamma_{m}\alpha^{s-1}_{m}},\\
		\end{cases}
	\end{displaymath}
    \end{small}
then the parity check matrix for this case can be given by 
\begin{small}
	\begin{displaymath}
		\begin{pmatrix}
				\frac{	\gamma_{1}}{\beta_{1}}f_{1}^{(1)}(\alpha_{1}, \beta_{1})&		\frac{	\gamma_{2}}{\beta_{2}}f_{1}^{(1)}(\alpha_{2}, \beta_{2})&\cdots&			\frac{	\gamma_{n}}{\beta_{n}}f_{1}^{(1)}(\alpha_{n}, \beta_{n})\\
                \hline\\
		& G(D,\omega, (y)+(n-k-1)O) &
		\end{pmatrix}
	\end{displaymath}
    \end{small}
where $\small{f_{1}^{(1)}(x, y)=x^{s-\ell-1}+\cdots+ a^{(1)}_{\frac{k-3}{2}-\ell}x^{s-\frac{k-1}{2}}+b^{(1)}_{0}x^{s-\ell-3}y+\cdots+b^{(1)}_{\frac{k-3}{2}-\ell}x^{s-\frac{k+3}{2}}y}$.
	
	\item The parity check matrix of $\mathcal{C}(D, kO,\ell, \eta)$ with even $k$: 
	
	Let $  e^{(1)}_{0}, e^{(1)}_{1}, \cdots, e^{(1)}_{\frac{k}{2}-\ell}, h^{(1)}_{0},h^{(1)}_{1},\cdots, h^{(1)}_{\frac{k}{2}-\ell}\in\mathbb{F}_{q}$ be defined by the following recursion
		\begin{displaymath}
		\begin{cases}
			e^{(1)}_{r-1}=-\frac{\sum_{m=1}^{n}\frac{\gamma_{m}}{\beta_{m}}\sum_{i=0}^{r-2}e_{i}^{(1)}\alpha^{s-\ell+\frac{k}{2}-i}_{m}+\sum_{m=1}^{n}\gamma_{m}\sum_{j=0}^{r}h_{j}^{(1)}\alpha^{s-\ell-1+\frac{k}{2}-j}_{m}}{\sum_{m=1}^{n}\frac{\gamma_{m}}{\beta_{m}}\alpha_{m}^{s-\ell+\frac{k}{2}-r+1}}\\
			h^{(1)}_{r}=-\frac{\sum_{m=1}^{n}\gamma_{m}\sum_{i=0}^{r-1}e_{i}^{(1)}\alpha^{s-\ell+\frac{k-4}{2}-i}_{m}+\sum_{m=1}^{n}\gamma_{m}\sum_{j=0}^{r-1}h_{j}^{(1)}\alpha^{s-\ell-1+\frac{k-4}{2}-j}_{m}\beta_{m}}{\sum_{m=1}^{n}\gamma_{m}\alpha_{m}^{s-\ell-1+\frac{k-4}{2}-r}\beta_{m}}\\
    \end{cases}
	\end{displaymath}
    for $	r=1, 2, \cdots, \frac{k}{2}-\ell$ with
    \begin{small}
	\begin{displaymath}
		\begin{cases}
			h^{(1)}_{0}=1\\
			e^{(1)}_{\frac{k}{2}-\ell}=-\frac{\sum_{m=1}^{n}\gamma_{m}\alpha^{s-1}_{m}+\eta\sum_{m=1}^{n}\gamma_{m}\sum^{\frac{k}{2}-\ell-1}_{i=0}e^{(1)}_{i}\alpha^{s-\ell-i+\frac{k-2}{2}}_{m}+\eta\sum^{n}_{m=1}\gamma_{m}\sum^{\frac{k}{2}-\ell-1}_{j=0}h_{j}^{(1)}\alpha^{s-\ell-1-j+\frac{k-2}{2}}_{m}\beta_{m}}{\eta\sum_{m=1}^{n}\gamma_{m}\alpha_{m}^{s-1}},\\
		\end{cases}
	\end{displaymath}
    \end{small}
	 then the parity check matrix for this case can be given by
     \begin{small}
	\begin{displaymath}
		\begin{pmatrix}
				\frac{\gamma_{1}}{\beta_{1}}	f_{2}^{(1)}(\alpha_{1}, \beta_{1})&\frac{\gamma_{2}}{\beta_{2}}f_{2}^{(1)}(\alpha_{2}, \beta_{2})&\cdots&\frac{\gamma_{n}}{\beta_{n}}f_{2}^{(1)}(\alpha_{n}, \beta_{n})\\
			\hline\\
            & G(D,\omega, (y)+(n-k-1)O)		\end{pmatrix}
	\end{displaymath}
    \end{small}
	where $f_{2}^{(1)}(x, y)=x^{s-\ell-1}y+\cdots+h^{(1)}_{\frac{k}{2}-\ell}x^{s-\frac{k}{2}-1}y+e^{(1)}_{0}x^{s-\ell}+\cdots+e^{(1)}_{\frac{k}{2}-\ell}x^{s-\frac{k}{2}}$.
    \end{enumerate}
	\label{7}
	\end{theorem}

\begin{proof} Notice that
\begin{displaymath}
	\mathcal{C}_{\mathcal{L}}(D, (2\ell+2)O)\subsetneq\mathcal{C}(D, kO,\ell, \eta)\subsetneq\mathcal{C}_{\mathcal{L}}(D, (k+1)O)
\end{displaymath}
for odd $k$ and
\begin{displaymath}
	\mathcal{C}_{\mathcal{L}}(D, (2\ell-1)O)\subsetneq\mathcal{C}(D, kO,\ell, \eta)\subsetneq\mathcal{C}_{\mathcal{L}}(D, (k+1)O)
\end{displaymath}
for even $k$, then we have 
\begin{displaymath}
	\mathcal{C}_{\mathcal{L}}(D, (k+1)O)^{\perp}\subsetneq\mathcal{C}(D, kO,\ell, \eta)^{\perp}\subsetneq\mathcal{C}_{\mathcal{L}}(D, (2\ell+2)O)^{\perp}
\end{displaymath}
and
\begin{displaymath}
	\mathcal{C}_{\mathcal{L}}(D, (k+1)O)^{\perp}\subsetneq\mathcal{C}(D, kO, \ell,\eta)^{\perp}\subsetneq\mathcal{C}_{\mathcal{L}}(D, (2\ell-1)O)^{\perp}.
\end{displaymath}

It can also be checked that $\mathcal{C}_{\mathcal{L}}(D, (k+1)O)^{\perp}$ has co-dimension 1 on $\mathcal{C}(D, kO, \ell,\eta)^{\perp}$, which allows us to consider a special kind of polynomial and the coefficients are to be determined.

Since $\mathcal{C}(D, kO,\ell,\eta)^{\perp}\subsetneq\mathcal{C}_{\mathcal{L}}(D, (2\ell+2)O)^{\perp}$, first, we consider the vector:
\begin{displaymath}
	(\frac{\gamma_{1}}{\beta_{1}}f_{1}^{(1)}(P_{1}), \frac{\gamma_{2}}{\beta_{2}}f_{1}^{(1)}(P_{2}),\cdots, \frac{\gamma_{n}}{\beta_{n}}f_{1}^{(1)}(P_{n})) 
\end{displaymath}
where $f^{(1)}_{1}(x, y)
=\sum_{i=0}^{\frac{k-3}{2}-\ell}a^{(1)}_{i}x^{s-\ell-1-i}+\sum_{j=0}^{\frac{k-3}{2}-\ell}b^{(1)}_{j}x^{s-\ell-3-j}y$.

 We claim that the vector does not belong to $\mathcal{C}_{\mathcal{L}}(D, (k+1)O)^{\perp}$, if not, there is a polynomial $g(x, y)=c_{0}x^{s-\frac{k+1}{2}} +c_{1}x^{s-\frac{k+1}{2}-1}+\cdots+c_{s-\frac{k+3}{2}}x+c_{s-\frac{k+1}{2}}+d_{0}x^{s-\frac{k+3}{2}}y+\cdots+d_{s-\frac{k+3}{2}}y$ such that $\frac{\gamma_{i}}{\beta_{i}}g(\alpha_{i}, \beta_{i})=\frac{\gamma_{i}}{\beta_{i}}f_{1}^{(1)}(\alpha_{i}, \beta_{i})$ for $i=1, 2, \cdots, n$, which implies that the polynomial $g(x, y)-f_{1}^{(1)}(x, y)$ has at least $n=2s$ different roots. On the other hand, it can be easily checked that $g(x, y)-f_{1}^{(1)}(x, y)\in\mathcal{L}((n-2)O)$. However, there are $n$ rational places such that $g(\alpha_{i}, \beta_{i})-f_{1}^{(1)}(\alpha_{i}, \beta_{i})=0$ for $i=1, 2, \cdots, n$, which is a contradiction.

The vector 
$(\frac{\gamma_{1}}{\beta_{1}}f_{1}^{(1)}(P_{1}), \frac{\gamma_{2}}{\beta_{2}}f_{1}^{(1)}(P_{2}),\cdots, \frac{\gamma_{n}}{\beta_{n}}f_{1}^{(1)}(P_{n}))$
belongs to $\mathcal{C}(D, kO,\ell, \eta)^{\perp}$ if and only if the following linear equation system holds

\begin{small}
\begin{displaymath}
	\begin{cases}
		\sum_{i=1}^{n}\frac{\gamma_{i}}{\beta_{i}}f_{1}^{(1)}(\alpha_{i},\beta_{i})\alpha_{i}^{\ell+2}= 0\\
		\cdots\\
		\sum_{i=1}^{n}\frac{\gamma_{i}}{\beta_{i}}f_{1}^{(1)}(\alpha_{i},\beta_{i})\alpha_{i}^{\frac{k-1}{2}}= 0\\
		\sum_{i=1}^{n}\frac{\gamma_{i}}{\beta_{i}}f_{1}^{(1)}(\alpha_{i},\beta_{i})\left( \beta_{i}\alpha_{i}^{\ell}+\eta\alpha_{i}^{\frac{k+1}{2}}\right) = 0\\
        \sum_{i=1}^{n}\gamma_{i}f_{1}^{(1)}(\alpha_{i},\beta_{i})\alpha_{i}^{\ell+1} = 0\\
        \cdots\\
		\sum_{i=1}^{n}\gamma_{i}f_{1}^{(1)}(\alpha_{i},\beta_{i})\alpha_{i}^{\frac{k-3}{2}}= 0,\\
	\end{cases}
\end{displaymath}
\end{small}
which is equivalent that
\begin{small}
\begin{displaymath}
	\begin{cases}
		a_{0}^{(1)}\sum_{m=1}^{n}\frac{\gamma_{m}}{\beta_{m}}\alpha_{m}^{s}+  
        a_{1}^{(1)}\sum_{m=1}^{n}\frac{\gamma_{m}}{\beta_{m}}\alpha_{m}^{s+1}        +b^{(1)}_{0}\sum_{m=1}^{n}\gamma_{m}\alpha_{m}^{s-1}= 0\\
		\cdots\\
		\sum_{m=1}^{n}\frac{\gamma_{m}}{\beta_{m}}\sum_{i=0}^{\frac{k-3}{2}-\ell}a^{(1)}_{i}\alpha_{m}^{s-\ell-1+\frac{k-1}{2}-i}+\sum_{m=1}^{n}\gamma_{m}\sum_{j=0}^{\frac{k-5}{2}-l}b^{(1)}_{j}\alpha_{m}^{s-\ell-3+\frac{k-1}{2}-j}= 0\\
\sum_{m=1}^{n}\gamma_{m}\alpha_{m}^{s-1}+\eta\sum_{m=1}^{n}\frac{\gamma_{m}}{\beta_{m}}\sum_{i=0}^{\frac{k-3}{2}-\ell}a^{(1)}_{i}\alpha_{m}^{s-\ell-1+\frac{k+1}{2}-i}
	+\eta\sum_{m=1}^{n}\gamma_{m}\sum_{j=0}^{\frac{k-3}{2}-\ell}b^{(1)}_{j}\alpha_{m}^{s-\ell-3+\frac{k+1}{2}-j}= 0\\
		
        a^{(1)}_{0}\sum^{n}_{m=1}\gamma_{m}\alpha_{m}^{s-1}+a^{(1)}_{1}\sum_{m=1}^{n}\gamma_{m}\alpha_{m}^{s}+b^{(1)}_{0}\sum^{n}_{m=1}\gamma_{m}\alpha^{s-2}_{m}\beta_{m}=0\\
        \cdots\\
        \sum_{m=1}^{n}\gamma_{m}\sum_{i=0}^{\frac{k-3}{2}-l}a^{(1)}_{i}\alpha_{m}^{s-\ell-1+\frac{k-3}{2}-i}+\sum_{m=1}^{n}\gamma_{m}\sum_{j=0}^{\frac{k-5}{2}-\ell}b^{(1)}_{j}\alpha_{m}^{s-\ell-3+\frac{k-3}{2}-j}\beta_{m}= 0.\\
	\end{cases}
\end{displaymath}
\end{small}
Notice that $a^{(1)}_{0}\neq 0$, so we can assume $a^{(1)}_{0}=1$ by linearity. If $a^{(1)}_{0}=0$, then by the linearity, we can also suppose $b^{(1)}_{0}=1$ if $\sum_{m=1}^{n}\gamma_{m}\alpha_{m}^{s-1}\neq 0$, otherwise, we have $b^{(1)}_{0}=0$, which directly leads to $a^{(1)}_{1}=0$ if $\sum_{m=1}^{n}\frac{\gamma_{m}}{\beta_{m}}\alpha_{m}^{s+1}\neq 0$ and the similar discussions can also be applied to the other equalities. As a consequence, we have $a^{(1)}_{i}= b^{(1)}_{i}=0$ for $i=0, 1, \cdots, \frac{k-3}{2}-\ell$ and $f^{(1)}_{1}(x, y)=0$ which contradicts the assumption that $f_{1}^{(1)}(x, y)$ is nonzero.

By the assumption $a^{(1)}_{0}=1$, we have a such recursive process, we can obtain all the coefficients of polynomial $f_{1}^{(1)}(x, y)$.

 For the TECC $\mathcal{C}(D, kO,\ell, \eta)$ with even dimension $k$, we consider the polynomial 
\begin{displaymath}
	\begin{split}
		f_{2}^{(1)}(x,y)
		&=\textstyle\sum_{i=0}^{\frac{k}{2}-\ell}e^{(1)}_{i}x^{s-\ell-i}+\sum_{j=0}^{\frac{k}{2}-\ell}h^{(1)}_{j}x^{s-\ell-1-j}y
	\end{split}
\end{displaymath}
then the codeword
\begin{displaymath}
 (\frac{\gamma_{1}}{\beta_{1}}f_{2}^{(1)}(P_{1}), \frac{\gamma_{2}}{\beta_{2}}f_{2}^{(1)}(P_{2}),\cdots,
\frac{\gamma_{n}}{\beta_{n}}f_{2}^{(1)}(P_{n})) 
\end{displaymath}
belongs to $\mathcal{C}(D, kO, \ell, \eta)^{\perp}$ if and only if the following linear equation system holds

\begin{small}
\begin{displaymath}
	\begin{cases}
		\sum_{i=1}^{n}\frac{\gamma_{i}}{\beta_{i}}f_{2}^{(1)}(\alpha_{i},\beta_{i}) \left( \alpha_{i}^{\ell}+\eta\beta_{i}\alpha_{i}^{\frac{k-2}{2}}\right)= 0\\
        \sum_{i=1}^{n}\frac{\gamma_{i}}{\beta_{i}}f_{2}^{(1)}(\alpha_{i},\beta_{i}) \alpha_{i}^{\ell+1}= 0\\
        \cdots\\
		\sum_{i=1}^{n}\frac{\gamma_{i}}{\beta_{i}}f_{2}^{(1)}(\alpha_{i},\beta_{i})\alpha_{i}^{\frac{k}{2}}= 0\\
		\sum_{i=1}^{n}\gamma_{i}f_{2}^{(1)}(\alpha_{i},\beta_{i})\alpha_{i}^{\ell-1}= 0\\
		\cdots\\
		\sum_{i=1}^{n}\gamma_{i}f_{2}^{(1)}(\alpha_{i},\beta_{i})\alpha_{i}^{\frac{k-4}{2}}= 0,\\
	\end{cases}
\end{displaymath}
\end{small}
which is equivalent that
\begin{small}
\begin{displaymath}
	\begin{cases}
		h^{(1)}_{0}\sum_{m=1}^{n}\gamma_{m}\alpha_{m}^{s-1}+
\eta\sum_{m=1}^{n}\gamma_{m}\sum_{i=0}^{\frac{k}{2}-\ell}e^{(1)}_{i}\alpha_{m}^{s-\ell-i+\frac{k-2}{2}}
+\eta\sum_{m=1}^{n}\gamma_{m}\sum_{j=0}^{\frac{k}{2}-\ell}h^{(1)}_{j}\alpha_{m}^{s-\ell-1-j+\frac{k-2}{2}}\beta_{m}=0\\
e_{0}^{(1)}\sum_{m=1}^{n}\frac{\gamma_{m}}{\beta_{m}}\alpha_{m}^{s+1}+	\sum_{m=1}^{n}\gamma_{m}h^{(1)}_{0}\alpha_{m}^{s-1}+\sum_{m=1}^{n}\gamma_{m}h^{(1)}_{1}\alpha_{m}^{s}=0\\
		\cdots\\
		\sum_{m=1}^{n}\frac{\gamma_{m}}{\beta_{m}}\sum_{i=0}^{\frac{k}{2}-\ell-1}e^{(1)}_{i}\alpha_{m}^{s-\ell+\frac{k}{2}-i}+	\sum_{m=1}^{n}\gamma_{m}\sum_{j=0}^{\frac{k}{2}-\ell}h^{(1)}_{j}\alpha_{m}^{s-\ell-1+\frac{k}{2}-j}=0\\
e^{(1)}_{0}\sum_{m=1}^{n}\gamma_{m}\alpha_{m}^{s-1}+	\sum_{m=1}^{n}\gamma_{m}h^{(1)}_{0}\alpha_{m}^{s-3}\beta_{m}+\sum_{m=1}^{n}\gamma_{m}h^{(1)}_{1}\alpha_{m}^{s-2}\beta_{m}=0\\
		\cdots\\
		\sum_{m=1}^{n}\gamma_{m}\sum_{i=0}^{\frac{k}{2}-\ell-1}e^{(1)}_{i}\alpha_{m}^{s-\ell-i+\frac{k-4}{2}}+	\sum_{m=1}^{n}\gamma_{m}\sum_{j=0}^{\frac{k}{2}-\ell}h^{(1)}_{j}\alpha_{m}^{s-\ell-1-j+\frac{k-4}{2}}\beta_{m}=0.\\
	\end{cases}
\end{displaymath}
\end{small}
Notice that $h^{(1)}_{0}\neq 0$, so we can assume $h^{(1)}_{0}=1$ by linearity. If $h^{(1)}_{0}=0$, then by the linearity we can suppose $e^{(1)}_{0}=1$ if $\sum_{m=1}^{n}\frac{\gamma_{m}}{\beta_{m}}\alpha_{m}^{s+1}\neq 0$, otherwise we have $e^{(1)}_{0}=0$, which directly leads to $h^{(1)}_{1}=0$ if $\sum_{m=1}^{n}\gamma_{m}\alpha_{m}^{s}\neq 0$ and the similar discussion can also be applied to the other equalities. As a consequence, we have $e^{(1)}_{i}= h^{(1)}_{i}=0$ for $i=0, 1, \cdots, \frac{k}{2}-\ell$ and $f^{(1)}_{2}(x, y)=0$ which contradicts the assumption that $f_{2}^{(1)}(x)$ is nonzero.

By assumption $h^{(1)}_{0}=1$, we have a recursive process, we can obtain all the coefficients of the polynomial $f_{2}^{(1)}(x, y)$.
\end{proof}
\begin{remark}
\vspace{-0.25 cm}
Each recursion $a^{(1)}_{r}, b^{(1)}_{r}, e^{(1)}_{r}, h^{(1)}_{r}$ for  $r=0, 1, \cdots, \frac{k-3}{2}-\ell$ is actually deduced from the linear equation systems in Theorem \ref{4}. After arranging the linear equation system by the orthogonality relations, we shall delete the parameters in $a^{(1)}_{r}, b^{(1)}_{r}, e^{(1)}_{r}, h^{(1)}_{r}$ whose coefficients in the linear equation systems are equal to zero, if not, we have the standard recursion process as in Theorem \ref{4}. For the other two types of TECCs, the recursion processes are also operated as the first type.
\vspace{-0.25 cm}
\end{remark}

By the similar discussions, we have the results for the TECC $\mathcal{C}(D, kO,\ell, \eta)$ constructed over the other two types of elliptic curves.

\begin{theorem} 
Notations as above. The parity check matrix of the TECC $\mathcal{C}(D, kO,\ell, \eta)$ constructed over the second type of elliptic curves can be classified as follows.
\begin{enumerate}
    \item 
	The parity check matrix of $\mathcal{C}(D, kO,\ell, \eta)$ for odd $k$:
	
		Let $a^{(2)}_{0},a_{1}^{(2)},\cdots, a^{(2)}_{\frac{k-3}{2}-\ell}, b^{(2)}_{0}, b^{(2)}_{1},\cdots, b^{(2)}_{\frac{k-3}{2}-\ell}\in\mathbb{F}_{q}$ be defined by the following recursion
	\begin{displaymath}
		\begin{cases}
			a^{(2)}_{r}=-\frac{\sum_{m=1}^{n}\gamma_{m}\sum_{i=0}^{r-1}a_{i}^{(2)}\alpha^{s-\ell-1+\frac{k-1}{2}-i}_{m}+\sum_{m=1}^{n}\gamma_{m}\sum_{j=0}^{r-1}b_{j}^{(2)}\alpha^{s-\ell-3+\frac{k-1}{2}-j}_{m}\beta_{m}}{\sum_{m=1}^{n}\gamma_{m}\alpha_{m}^{s-\ell-1+\frac{k-1}{2}-r}}\\
			b^{(2)}_{r-1}=-\frac{\sum_{m=1}^{n}\gamma_{m}\sum_{i=0}^{r}a_{i}^{(2)}\alpha^{s-\ell-1+\frac{k-3}{2}-i}_{m}\beta_{m}+\sum_{m=1}^{n}\gamma_{m}\sum_{j=0}^{r-2}b_{j}^{(2)}\alpha^{s-\ell-3+\frac{k-3}{2}-j}_{m}\beta^{2}_{m}}{\sum_{m=1}^{n}\gamma_{m}\alpha_{m}^{s-\ell-3+\frac{k-3}{2}-r}\beta^{2}_{m}}\\
        	\end{cases}
	\end{displaymath}
    for $r=1, 2, \cdots, \frac{k-3}{2}-\ell$ with
    \begin{small}
		\begin{displaymath}
		\begin{cases}
			a^{(2)}_{0}=1\\
			b^{(2)}_{\frac{k-3}{2}-\ell}=-\frac{\sum_{m=1}^{n}\gamma_{m}\alpha^{s-1}_{m}\beta_{m}+\eta\sum_{m=1}^{n}\gamma_{m}\sum^{\frac{k-3}{2}-\ell}_{i=0}a^{(2)}_{i}\alpha^{s-\ell-1+\frac{k+1}{2}-i}_{m}+\eta\sum_{m=1}^{n}\gamma_{m}\sum^{\frac{k-3}{2}-\ell-1}_{j=0}b_{j}^{(2)}\alpha^{s-\ell-3+\frac{k+1}{2}-j}_{m}\beta_{m}}{\eta\sum_{m=1}^{n}\gamma_{m}\alpha^{s-1}_{m}\beta_{m}},\\
		\end{cases}
	\end{displaymath}
    \end{small}
    then the parity check matrix for this case can be given by 
    \begin{small}
	\begin{displaymath}
			\begin{pmatrix}
				\gamma_{1}f_{1}^{(2)}(\alpha_{1}, \beta_{1})&	\gamma_{2}f_{1}^{(2)}(\alpha_{2}, \beta_{2})&\cdots&\gamma_{n}f_{1}^{(2)}(\alpha_{n}, \beta_{n})\\
			\hline\\
             &G(D,\omega, (n-k-1)O)		
             \end{pmatrix}
	\end{displaymath}
    \end{small}
	where $f_{1}^{(2)}(x, y)= x^{s-\ell-1}+\cdots+ a^{(2)}_{\frac{k-3}{2}-\ell}x^{s-\frac{k-1}{2}}+b^{(2)}_{0}x^{s-\ell-3}y+\cdots+b^{(2)}_{\frac{k-3}{2}-\ell}x^{s-\frac{k+3}{2}}y$.
	
	\item The parity check matrix of $\mathcal{C}(D, kO,\ell, \eta)$ for even $k$:
	
		Let $  e^{(2)}_{0},e^{(2)}_{1},\cdots, e^{(2)}_{\frac{k}{2}-\ell}, h^{(2)}_{0},h^{(2)}_{1},\cdots, h^{(2)}_{\frac{k}{2}-\ell}\in\mathbb{F}_{q}$ be defined by the following recursion
	\begin{displaymath}
		\begin{cases}
			e^{(2)}_{r-1}=-\frac{\sum_{m=1}^{n}\gamma_{m}\sum_{i=0}^{r-2}e_{i}^{(2)}\alpha^{s-\ell+\frac{k}{2}-i}_{m}+\sum_{m=1}^{n}\gamma_{m}\sum_{j=0}^{r}h_{j}^{(2)}\alpha^{s-\ell-1+\frac{k}{2}-j}_{m}\beta_{m}}{\sum_{m=1}^{n}\gamma_{m}\alpha_{m}^{s-\ell+\frac{k}{2}-r+1}}\\
			h^{(2)}_{r}=-\frac{\sum_{m=1}^{n}\gamma_{m}\sum_{i=0}^{r-1}e_{i}^{(2)}\alpha^{s-\ell+\frac{k-4}{2}-i}_{m}\beta_{m}+\sum_{m=1}^{n}\gamma_{m}\sum_{j=0}^{r-1}h_{j}^{(2)}\alpha^{s-\ell-1+\frac{k-4}{2}-j}_{m}\beta^{2}_{m}}{\sum_{m=1}^{n}\gamma_{m}\alpha_{m}^{s-\ell-1+\frac{k-4}{2}-r}\beta^{2}_{m}}\\
    \end{cases}
	\end{displaymath}
    for $r=1, 2, \cdots, \frac{k}{2}-\ell$ with
    \begin{small}
	\begin{displaymath}
		\begin{cases}
			h^{(2)}_{0}=1\\
			e^{(2)}_{\frac{k}{2}-\ell}=-\frac{\sum_{m=1}^{n}\gamma_{m}\alpha^{s-1}_{m}\beta_{m}+\eta\sum_{m=1}^{n}\gamma_{m}\sum^{\frac{k}{2}-\ell-1}_{i=0}e^{(2)}_{i}\alpha^{s-\ell-i+\frac{k-2}{2}}_{m}\beta_{m}+\eta\sum^{n}_{m=1}\gamma_{m}\sum^{\frac{k}{2}-\ell-1}_{j=0}h_{j}^{(2)}\alpha^{s-\ell-1-j+\frac{k-2}{2}}_{m}\beta^{2}_{m}}{\eta\sum_{m=1}^{n}\gamma_{m}\alpha_{m}^{s-1}\beta_{m}},\\
		\end{cases}
	\end{displaymath}
    \end{small}
    then the parity check matrix for this case can be given by 
    \begin{small}
	\begin{displaymath}
		\begin{pmatrix}
				\gamma_{1}f_{2}^{(2)}(\alpha_{1}, \beta_{1})&	\gamma_{2}f_{2}^{(2)}(\alpha_{2}, \beta_{2})&\cdots&\gamma_{n}f_{2}^{(2)}(\alpha_{n}, \beta_{n})\\
			\hline\\
             &G(D,\omega, (n-k-1)O)\\		\end{pmatrix}
		\end{displaymath}
        \end{small}
	where $f_{2}^{(2)}(x, y)=x^{s-\ell-1}y+\cdots+h^{(2)}_{\frac{k}{2}-\ell}x^{s-\frac{k}{2}-1}y+e^{(2)}_{0}x^{s-\ell}+\cdots+e^{(2)}_{\frac{k}{2}-\ell}x^{s-\frac{k}{2}}$.
    \end{enumerate}
    \label{55}
\end{theorem}

\begin{theorem}
Notations as above. The parity check matrix of the TECC $\mathcal{C}(D, kO,\ell, \eta)$ constructed over the third type of elliptic curves can be classified as follows.
\begin{enumerate}
    \item The parity check matrix of $\mathcal{C}(D, kO,\ell, \eta)$ for odd $k$:
		
		Let $a^{(3)}_{0}, a_{1}^{(3)}\cdots, a^{(3)}_{\frac{k-3}{2}-\ell}, b^{(3)}_{0}, b^{(3)}_{1},\cdots, b^{(3)}_{\frac{k-3}{2}-\ell}\in\mathbb{F}_{q}$ be defined by the following recursion
	\begin{displaymath}
		\begin{cases}
			a^{(3)}_{r}=-\frac{\sum_{m=1}^{n}\frac{\gamma_{m}}{a\alpha_{m}+b}\sum_{i=0}^{r-1}a_{i}^{(3)}\alpha^{s-\ell-1+\frac{k-1}{2}-i}_{m}+\sum_{m=1}^{n}\gamma_{m}\sum_{j=0}^{r-1}b_{j}^{(3)}\alpha^{s-\ell-3+\frac{k-1}{2}-j}_{m}\beta_{m}}{\sum_{m=1}^{n}\frac{\gamma_{m}}{a\alpha_{m}+b}\alpha_{m}^{s-\ell-1+\frac{k-1}{2}-r}}\\
			b^{(3)}_{r-1}=-\frac{\sum_{m=1}^{n}\gamma_{m}\sum_{i=0}^{r}a_{i}^{(3)}\alpha^{s-\ell-1+\frac{k-3}{2}-i}_{m}\beta_{m}+\sum_{m=1}^{n}\gamma_{m}(a\alpha_{m}+b)\sum_{j=0}^{r-2}b_{j}^{(3)}\alpha^{s-\ell-3+\frac{k-3}{2}-j}_{m}\beta^{2}_{m}}{\sum_{m=1}^{n}\gamma_{m}(a\alpha_{m}+b)\alpha_{m}^{s-\ell-3+\frac{k-3}{2}-r}\beta^{2}_{m}}\\
        	\end{cases}
	\end{displaymath}
    for
    $r=1, 2, \cdots, \frac{k-3}{2}-\ell$ with
		\begin{displaymath}
		\begin{cases}
			a^{(3)}_{0}=1\\
			b^{(3)}_{\frac{k-3}{2}-\ell}=-\frac{\sum\limits_{m=1}^{n}\frac{\gamma_{m}}{a\alpha_{m}+b}\alpha^{s-1}_{m}+\eta\sum\limits_{m=1}^{n}\frac{\gamma_{m}}{a\alpha_{m}+b}\sum\limits^{\frac{k-3}{2}-\ell}_{i=0}a^{(3)}_{i}\alpha^{s-\ell-1+\frac{k+1}{2}-i}_{m}+\eta\sum\limits_{m=1}^{n}\gamma_{m}\sum\limits^{\frac{k-3}{2}-\ell-1}_{j=0}b_{j}^{(3)}\alpha^{s-\ell-3+\frac{k+1}{2}-j}_{m}\beta_{m}}{\eta\sum\limits_{m=1}^{n}\frac{\gamma_{m}}{a\alpha_{m}+b}\alpha^{s-1}_{m}},\\
		\end{cases}
	\end{displaymath}	
    then the parity check matrix for this case can be given by 
    \begin{small}
		\begin{displaymath}
			\begin{pmatrix}
			\frac{\gamma_{1}}{a\alpha_{1}+b}f_{1}^{(3)}(\alpha_{1}, \beta_{1})&	\frac{\gamma_{2}}{a\alpha_{2}+b}f_{1}^{(3)}(\alpha_{2}, \beta_{2})&\cdots&		\frac{\gamma_{n}}{a\alpha_{n}+b}f_{1}^{(3)}(\alpha_{n}, \beta_{n})\\
			\hline\\
             & G(D,\omega, (ax+b)+(n-k-1)O)			\end{pmatrix}
		\end{displaymath}
        \end{small}
        where $
		    f_{1}^{(3)}(x, y)=x^{s-\ell-1}+\cdots+ a^{(3)}_{\frac{k-3}{2}-\ell}x^{s-\frac{k-1}{2}}+b^{(3)}_{0}x^{s-\ell-3}(ax+b)y+\cdots+ b^{(3)}_{\frac{k-3}{2}-\ell}x^{s-\frac{k+3}{2}}
         (ax+b)y$.
		
		\item The parity check matrix of $\mathcal{C}(D, kO,\ell, \eta)$ for even $k$:
		
		Let $  e^{(3)}_{0},e^{(3)}_{1},\cdots, e^{(3)}_{\frac{k}{2}-\ell}, h^{(3)}_{0},h^{(3)}_{1},\cdots, h^{(3)}_{\frac{k}{2}-\ell}\in\mathbb{F}_{q}$ be defined by the following recursion
	\begin{displaymath}
		\begin{cases}
			e^{(3)}_{r-1}=-\frac{\sum_{m=1}^{n}\frac{\gamma_{m}}{a\alpha_{m}+b}\sum_{i=0}^{r-2}e_{i}^{(1)}\alpha^{s-\ell+\frac{k}{2}-i}_{m}+\sum_{m=1}^{n}\gamma_{m}\sum_{j=0}^{r}h_{j}^{(1)}\alpha^{s-\ell-1+\frac{k}{2}-j}_{m}\beta_{m}}{\sum_{m=1}^{n}\frac{\gamma_{m}}{a\alpha_{m}+b}\alpha_{m}^{s-\ell+\frac{k}{2}-r+1}}\\
			h^{(3)}_{r}=-\frac{\sum_{m=1}^{n}\gamma_{m}\sum_{i=0}^{r-1}e_{i}^{(1)}\alpha^{s-\ell+\frac{k-4}{2}-i}_{m}\beta_{m}+\sum_{m=1}^{n}\gamma_{m}(a\alpha_{m}+b)\sum_{j=0}^{r-1}h_{j}^{(1)}\alpha^{s-\ell-1+\frac{k-4}{2}-j}_{m}\beta^{2}_{m}}{\sum_{m=1}^{n}\gamma_{m}(a\alpha_{m}+b)\alpha_{m}^{s-\ell-1+\frac{k-4}{2}-r}\beta^{2}_{m}}\\
    \end{cases}
	\end{displaymath}
    for $	r=1, 2, \cdots, \frac{k}{2}-\ell$ with
    \begin{small}
	\begin{displaymath}
		\begin{cases}
			h^{(3)}_{0}=1\\
			e^{(3)}_{\frac{k}{2}-\ell}=-\frac{\sum\limits_{m=1}^{n}\gamma_{m}\alpha^{s-1}_{m}\beta_{m}+\eta\sum\limits_{m=1}^{n}\gamma_{m}\sum\limits^{\frac{k}{2}-\ell-1}_{i=0}e^{(3)}_{i}\alpha^{s-\ell-i+\frac{k-2}{2}}_{m}\beta_{m}+\eta\sum\limits^{n}_{m=1}\gamma_{m}(a\alpha_{m}+b)\sum\limits^{\frac{k}{2}-\ell-1}_{j=0}h_{j}^{(3)}\alpha^{s-\ell-1-j+\frac{k-2}{2}}_{m}\beta^{2}_{m}}{\eta\sum\limits_{m=1}^{n}\gamma_{m}\alpha_{m}^{s-1}\beta_{m}},\\
		\end{cases}
	\end{displaymath}
    \end{small}
    then the parity check matrix for this case can be given by 
    \begin{small}
		\begin{displaymath}
			\begin{pmatrix}
				\frac{\gamma_{1}}{a\alpha_{1}+b}	f_{2}^{(3)}(\alpha_{1}, \beta_{1})&\frac{\gamma_{2}}{a\alpha_{2}+b}f_{2}^{(3)}(\alpha_{2}, \beta_{2})&\cdots&\frac{\gamma_{n}}{a\alpha_{n}+b}f_{2}^{(3)}(\alpha_{n}, \beta_{n})\\
				\hline\\
                & G(D,\omega, (ax+b)+(n-k-1)O)			
                 \end{pmatrix}
		\end{displaymath}
        \end{small}
		where $f_{2}^{(3)}(x, y)=x^{s-\ell-1}(ax+b)y+\cdots+h^{(3)}_{\frac{k}{2}-\ell}x^{s-\frac{k}{2}-1}(ax+b)y+e^{(3)}_{0}x^{s-\ell}+\cdots+e^{(3)}_{\frac{k}{2}-\ell}x^{s-\frac{k}{2}}$.
	\end{enumerate}
    \label{56}
\end{theorem}

To construct self-dual codes from TECCs, we need to consider a generalized version of TECCs with scaling coefficients. 



\begin{definition}
\vspace{-0.25 cm}
Let $P_{1}, \cdots, P_{n}, O$ be distinct $\mathbb{F}_{q}$-rational points on elliptic curve $E$. Put $D=\sum^{n}_{i=1}P_i$. The twisted generalized elliptic curve codes (TGECCs) $\mathcal{C}(D, kO, \eta, \mathbf{v})$ is defined to be
\begin{displaymath}
\mathcal{C}(D, kO,\ell, \eta, \mathbf{v}):=\{(v_{1}f(P_{1}), v_{2}f(P_{2}), \cdots, v_{n}f(P_{n})): f\in S_{\ell}\}
\end{displaymath}
where $\mathbf{v}=(v_{1}, v_{2}, \cdots, v_{n})\in\mathbb{F}^{n}_{q}$ with $v_{i}\neq 0$ for $1\leq i\leq n$.
\vspace{-0.25 cm}
\end{definition}

\begin{remark}
    We can also define a generalized version of ECCs by substituting the Riemann-Roch spaces ($\mathcal{L}(kO)$ et. al.) for the defining sets of TECCs as the operations in \cite{10} and such constructions are called the generalized elliptic curve codes (GECCs). 
\end{remark}

Now we shall assume $\ell=\frac{k-3}{2}$ for odd $k$ and $\ell=\frac{k}{2}$ for even $k$, respectively in the following of this section.

\begin{corollary}\label{101}
		The parity-check matrices of TGECCs $\mathcal{C}(D, kO,\ell, \eta, \mathbf{v})$ constructed over the first type of elliptic curves can be classified as follows:
		\begin{enumerate}
			\item For odd $k$, 
			the parity-check matrix of $\mathcal{C}(D, kO,\ell, \eta, \mathbf{v})$ can be given by:
			\begin{small}
				\begin{displaymath}
					\begin{pmatrix}
						\frac{	\gamma_{1}}{v_{1}\beta_{1}}f_{1}^{(1)}(\alpha_{1}, \beta_{1})&		\frac{	\gamma_{2}}{v_{2}\beta_{2}}f_{1}^{(1)}(\alpha_{2}, \beta_{2})&\cdots&			\frac{	\gamma_{n}}{v_{n}\beta_{n}}f_{1}^{(1)}(\alpha_{n}, \beta_{n})\\
						\hline\\
						& \mathbf{\frac{1}{v}}\star G(D,\omega, (y)+(n-k-1)O) &
					\end{pmatrix}
				\end{displaymath}
			\end{small}
			where 
			\begin{small}
				\begin{displaymath}
					f^{(1)}_{1}(x, y)=            x^{s-\frac{k-1}{2}}-\frac{\sum_{i=1}^{n}\gamma_{i}\alpha_{i}^{s-1}+\eta\sum_{i=1}^{n}\frac{\gamma_{i}}{\beta_{i}}\alpha_{i}^{s+1}}{\eta\sum_{i=1}^{n}\gamma_{i}\alpha_{i}^{s-1}}x^{s-\frac{k+3}{2}}y.
				\end{displaymath}
			\end{small}
			\item For even $k$, the parity-check matrix of $\mathcal{C}(D, kO,\ell, \eta, \mathbf{v})$ can be given by:
			\begin{small}
				\begin{displaymath}
					\begin{pmatrix}
						\frac{\gamma_{1}}{v_{1}\beta_{1}}	f_{2}^{(1)}(\alpha_{1}, \beta_{1})&\frac{\gamma_{2}}{v_{2}\beta_{2}}f_{2}^{(1)}(\alpha_{2}, \beta_{2})&\cdots&\frac{\gamma_{n}}{v_{n}\beta_{n}}f_{2}^{(1)}(\alpha_{n}, \beta_{n})\\
						\hline\\
						& \mathbf{\frac{1}{v}}\star G(D,\omega, (y)+(n-k-1)O) &
					\end{pmatrix}
				\end{displaymath}
			\end{small}
			where 
			\begin{small}
				\begin{displaymath}
					f^{(1)}_{2}(x, y)=
					x^{s-\frac{k}{2}-1}y-\frac{\sum_{i=1}^{n}\gamma_{i}\alpha_{i}^{s-1}+\eta\sum_{i=1}^{n}\gamma_{i}\alpha_{i}^{s-2}\beta_{i}}{\eta\sum_{i=1}^{n}\gamma_{i}\alpha_{i}^{s-1}}x^{s-\frac{k}{2}}.
				\end{displaymath}
			\end{small}
		\end{enumerate}
	\end{corollary}
	\begin{remark}
		Note that the sum $\sum_{i=1}^{n}\gamma_{i}\alpha_{i}^{s-1}\neq 0$, otherwise, we shall obtain $\mathcal{C}_{\mathcal{L}}(D,kO)$ has a $n-k+1$ dimensional parity-check matrix which is a contradiction to the residue theorem.
	\end{remark}

\begin{corollary}
		The parity-check matrices of TGECCs $\mathcal{C}(D, kO,\ell, \eta, \mathbf{v})$ constructed over the second type of elliptic curves can be classified as follows:
		\begin{enumerate}
			\item For odd $k$, 
			the parity-check matrix of $\mathcal{C}(D, kO,\ell, \eta, \mathbf{v})$ can be given by:
			\begin{small}
				\begin{displaymath}
					\begin{pmatrix}
						\frac{	\gamma_{1}}{v_{1}}f_{1}^{(1)}(\alpha_{1}, \beta_{1})&		\frac{	\gamma_{2}}{v_{2}}f_{1}^{(1)}(\alpha_{2}, \beta_{2})&\cdots&			\frac{	\gamma_{n}}{v_{n}}f_{1}^{(1)}(\alpha_{n}, \beta_{n})\\
						\hline\\
						& \mathbf{\frac{1}{v}}\star G(D,\omega, (n-k-1)O) &
					\end{pmatrix}
				\end{displaymath}
			\end{small}
			where 
			 \begin{small}
     \begin{displaymath}
     f^{(2)}_{1}(x, y)=
           \eta\sum\limits_{i=1}^{n}\gamma_{i}\alpha_{i}^{s-1}\beta_{i}x^{s-\frac{k-1}{2}}+\left(\sum\limits_{i=1}^{n}\gamma_{i}\alpha_{i}^{s-1}\beta_{i}+\eta\sum\limits_{i=1}^{n}\gamma_{i}\alpha_{i}^{s+1}\right)x^{s-\frac{k+3}{2}}y.
           \end{displaymath}
           \end{small}
			\item For even $k$, the parity-check matrix of $\mathcal{C}(D, kO,\ell, \eta, \mathbf{v})$ can be given by:
			\begin{small}
				\begin{displaymath}
					\begin{pmatrix}
						\frac{\gamma_{1}}{v_{1}\beta_{1}}	f_{2}^{(1)}(\alpha_{1}, \beta_{1})&\frac{\gamma_{2}}{v_{2}\beta_{2}}f_{2}^{(1)}(\alpha_{2}, \beta_{2})&\cdots&\frac{\gamma_{n}}{v_{n}\beta_{n}}f_{2}^{(1)}(\alpha_{n}, \beta_{n})\\
						\hline\\
						& \mathbf{\frac{1}{v}}\star G(D,\omega, (y)+(n-k-1)O) &
					\end{pmatrix}
				\end{displaymath}
			\end{small}
			\begin{small}
    \begin{displaymath}
    f^{(2)}_{2}(x, y)=
\eta\sum\limits_{i=1}^{n}\gamma_{i}\alpha_{i}^{s-1}\beta_{i}x^{s-\frac{k}{2}-1}y+\left(\sum\limits_{i=1}^{n}\gamma_{i}\alpha_{i}^{s-1}\beta_{i}+\eta\sum\limits_{i=1}^{n}\gamma_{i}\alpha_{i}^{s-2}\beta^{2}_{i}\right)x^{s-\frac{k}{2}}.
\end{displaymath}
\end{small}
		\end{enumerate}
	\end{corollary}

\begin{corollary}
		The parity-check matrices of TGECCs $\mathcal{C}(D, kO,\ell, \eta, \mathbf{v})$ constructed over the first type of elliptic curves can be classified as follows:
		\begin{enumerate}
			\item For odd $k$, 
			the parity-check matrix of $\mathcal{C}(D, kO,\ell, \eta, \mathbf{v})$ can be given by:
			\begin{small}
				\begin{displaymath}
					\begin{pmatrix}
						\frac{	\gamma_{1}}{v_{1}(a\alpha_{1}+b)}f_{1}^{(1)}(\alpha_{1}, \beta_{1})&		\frac{	\gamma_{2}}{v_{2}(a\alpha_{2}+b)}f_{1}^{(1)}(\alpha_{2}, \beta_{2})&\cdots&			\frac{	\gamma_{n}}{v_{n}(a\alpha_{n}+b)}f_{1}^{(1)}(\alpha_{n}, \beta_{n})\\
						\hline\\
						& \mathbf{\frac{1}{v}}\star G(D,\omega, (ax+b)+(n-k-1)O) &
					\end{pmatrix}
				\end{displaymath}
			\end{small}
			 where 
     \begin{small}
     \begin{displaymath}
     f^{(3)}_{1}(x, y)=
\eta\sum\limits_{i=1}^{n}\gamma_{i}\alpha_{i}^{s-1}\beta_{i}x^{s-\frac{k-1}{2}}+\left(\sum\limits_{i=1}^{n}\gamma_{i}\alpha_{i}^{s-1}\beta_{i}+\eta\sum\limits_{i=1}^{n}\frac{\gamma_{i}\alpha^{s+1}_{i}}{a\alpha_{i}+b}\right)  
x^{s-\frac{k+3}{2}}
(ax+b)y.
\end{displaymath}
\end{small}
			\item For even $k$, the parity-check matrix of $\mathcal{C}(D, kO,\ell, \eta, \mathbf{v})$ can be given by:
			\begin{small}
				\begin{displaymath}
					\begin{pmatrix}
						\frac{\gamma_{1}}{v_{1}\beta_{1}}	f_{2}^{(1)}(\alpha_{1}, \beta_{1})&\frac{\gamma_{2}}{v_{2}\beta_{2}}f_{2}^{(1)}(\alpha_{2}, \beta_{2})&\cdots&\frac{\gamma_{n}}{v_{n}\beta_{n}}f_{2}^{(1)}(\alpha_{n}, \beta_{n})\\
						\hline\\
						& \mathbf{\frac{1}{v}}\star G(D,\omega, (y)+(n-k-1)O) &
					\end{pmatrix}
				\end{displaymath}
			\end{small}
			where 
    \begin{small}
    \begin{displaymath}
    f^{(3)}_{2}(x, y)=
\eta\sum\limits_{i=1}^{n}\gamma_{i}\alpha_{i}^{s-1}\beta_{i}x^{s-\frac{k}{2}-1}(ax+b)y+\left(\sum\limits_{i=1}^{n}\gamma_{i}\alpha_{i}^{s-1}\beta_{i}+\eta\sum\limits_{i=1}^{n}\gamma_{i}(a\alpha_{i}+b)\alpha_{i}^{s-2}\beta_{i}\right)x^{s-\frac{k}{2}}.
\end{displaymath}
\end{small}
			\label{100}
		\end{enumerate}
	\end{corollary}

\begin{lemma}\label{111}
    Let $n=2k$ ($k\geq 3$). Let $G$ and $H$ be the generator matrix and parity-check matrix of TECC $\mathcal{C}(D, kO, \ell,\eta,\mathbf{v})$ with $\ell=\frac{k-3}{2}$ for odd $k$ or $\frac{k}{2}$ for even $k$, respectively. Let $\mathbf{g}_{i}$ and $\mathbf{h}_{i}$ denote the $i$-th row of $G$ and $H$, respectively. For $\eta\in\mathbb{F}^{*}_{q}$, then $\mathrm{Span}\{\mathbf{g}_{0}, \mathbf{g}_{1}, \cdots, \mathbf{g}_{k-1}\}=\mathrm{Span}\{\mathbf{h}_{0}, \mathbf{h}_{1}, \cdots, \mathbf{h}_{k-1}\}$
    if and only if the following conditions hold:
    \begin{enumerate}
        \item $\mathrm{Span}\{\mathbf{g}_{0}, \mathbf{g}_{1}, \cdots, \mathbf{g}_{k-2}\}=\mathrm{Span}\{\mathbf{h}_{1}, \mathbf{h}_{1}, \cdots, \mathbf{h}_{k-1}\};$
        \item $\mathbf{g}_{k-1}=a\cdot\mathbf{h}_{0}$ for some $a\in \mathbb{F}^{*}_{q}.$
    \end{enumerate}
    
    \begin{proof} Without loss of the generality, we only prove for the first type of TECC $\mathcal{C}(D, kO,\frac{k-3}{2}, \eta, \mathbf{v})$. It is obvious that we only need to show the necessary part. 

    First, we give the proof for the case of dimension $k\geq 5$. 
    
    Since the two vectors $\mathbf{g}_{0}, \mathbf{g}_{\frac{k-3}{2}}\in\mathrm{Span}\{\mathbf{h}_{0}, \mathbf{h}_{1}, \cdots, \mathbf{h}_{k-1}\}$, we have
    \begin{itemize}
	\item there exists $ c_{0}, c_{1}, c_{2}, \cdots, c_{k-1}\in\mathbb{F}_{q}$ such that
    \begin{small}
	\begin{displaymath}
		\frac{v^{2}_{i}\beta_{i}}{\gamma_{i}}=c_{0}+c_{1}\alpha_{i}+\cdots+c_{\frac{k-1}{2}}\alpha_{i}^{\frac{k-1}{2}}+c_{\frac{k+1}{2}}\beta_{i}+\cdots+c_{k-2}\alpha^{\frac{k-5}{2}}_{i}\beta_{i}+c_{k-1}f^{(1)}_{1}(\alpha_{i}, \beta_{i})
        \end{displaymath}
        \end{small}
	for any $i= 1, 2, \cdots, n$.
	\item there exists $ d_{0}, d_{1}, d_{2}, \cdots, d_{k-1}\in\mathbb{F}_{q}$ such that
    \begin{small}
	\begin{displaymath}
		\frac{v^{2}_{i}\beta_{i}}{\gamma_{i}}\alpha^{\frac{k-3}{2}}_{i}=d_{0}+d_{1}\alpha_{i}+\cdots+d_{\frac{k-1}{2}}\alpha_{i}^{\frac{k-1}{2}}+d_{\frac{k+1}{2}}\beta_{i}+\cdots+d_{k-2}\alpha^{\frac{k-5}{2}}_{i}\beta_{i}+d_{k-1}f^{(1)}_{1}(\alpha_{i}, \beta_{i})
        \end{displaymath}
        \end{small}
	for any $i=1, 2, \cdots, n$.
\end{itemize}

Consider the polynomials $h_{c}(x, y)=c_{0}+c_{1}x+\cdots+c_{k-2}x^{\frac{k-5}{2}}y+c_{k-1}f^{(1)}_{1}(x, y)$ and $h_{d}(x, y)=d_{0}+d_{1}x+\cdots+d_{k-2}x^{\frac{k-5}{2}}y+d_{k-1}f^{(1)}_{1}(x, y)\in\mathbb{F}_{q}[x, y]$. Then the polynomial $ h_{c}(x, y)x^{\frac{k-3}{2}}-h_{d}(x, y)$ has $n$ different roots $\{(\alpha_{i}, \beta_{i}),\,i=1, 2, \cdots, n\}$. But $h_{c}(x, y)x^{\frac{k-3}{2}}-h_{d}(x, y)\in\mathcal{L}((n-2)O)$, which means $h_{c}(x, y)x^{\frac{k-3}{2}}-h_{d}(x, y)\equiv 0$. 
Comparing the coefficients, we obtain: 
 $$\begin{cases}
 d_{0}=\cdots=d_{\frac{k-5}{2}}=d_{\frac{k+1}{2}}=\cdots=d_{k-2}=0\\
 c_{3}=\cdots=c_{\frac{k-1}{2}}=c_{\frac{k+3}{2}}=\cdots=c_{k-1}=0
 \end{cases}
 $$ and the following equalities:
\begin{displaymath}
\begin{cases}
	 d_{\frac{k-3}{2}}=c_{0}\\
     d_{\frac{k-1}{2}}=c_{1}\\
     d_{k-1}=c_{\frac{k+1}{2}}\\
     \eta d_{k-1}=c_{2}.
\end{cases}
\end{displaymath}
Then we have $h_{c}(x, y)=c_{0}+c_{1}x+c_{\frac{k+1}{2}}(y+\eta x^{2})$.

Since $\mathbf{g}_{\frac{k-1}{2}}\in\mathrm{Span}\{\mathbf{h}_{0}, \mathbf{h}_{1}, \cdots, \mathbf{h}_{k-1}\}$, there exists a polynomial $h_{e}(x)=e_{0}+e_{1}x+\cdots+e_{k-1}f^{(1)}_{1}(x, y)\in\mathbb{F}_{q}[x, y]$ such that
    \begin{small}
	\begin{displaymath}
		\frac{v^{2}_{i}\beta_{i}}{\gamma_{i}}\alpha^{\frac{k-1}{2}}_{i}=e_{0}+e_{1}\alpha_{i}+\cdots+e_{\frac{k-1}{2}}\alpha_{i}^{\frac{k-1}{2}}+e_{\frac{k+1}{2}}\beta_{i}+\cdots+e_{k-2}\alpha^{\frac{k-5}{2}}_{i}\beta_{i}+e_{k-1}f^{(1)}_{1}(\alpha_{i}, \beta_{i})
        \end{displaymath}
        \end{small}
	for any $i=1, 2, \cdots, n$. The polynomial $h_{c}(x, y)x^{\frac{k-1}{2}}-h_{e}(x,y)$ has $n$ different roots $\{(\alpha_{i}, \beta_{i}),\,i=1, 2, \cdots, n\}$. But $h_{c}(x, y)x^{\frac{k-1}{2}}-h_{e}(x,y)\in\mathcal{L}((k+3)O)\subseteq\mathcal{L}((n-1)O)$. So we have
    \begin{displaymath}
        h_{c}(x, y)x^{\frac{k-1}{2}}-h_{e}(x,y)\equiv 0.
    \end{displaymath}
Hence, by comparing the coefficients, we obtain $h_{c}(x, y)=c_{0}=e_{\frac{k-1}{2}}\in\mathbb{F}^{*}_{q}$, which means 
\begin{displaymath}
\frac{v^{2}_{i}\beta_{i}}{\gamma_{i}}=c_{0}\quad \text{for all } i=1, 2, \cdots, n.
\end{displaymath}
By plugging the above relation in the generator matrix, we get 
\[
\mathrm{Span}\{\mathbf{g}_{0}, \mathbf{g}_{1}, \cdots, \mathbf{g}_{k-2}\}=\mathrm{Span}\{\mathbf{h}_{1}, \mathbf{h}_{1}, \cdots, \mathbf{h}_{k-1}\}.
\]
Combining the above equality with $\mathrm{Span}\{\mathbf{g}_{0}, \mathbf{g}_{1}, \cdots, \mathbf{g}_{k-1}\}=\mathrm{Span}\{\mathbf{h}_{0}, \mathbf{h}_{1}, \cdots, \mathbf{h}_{k-1}\}$, there exists some $a\in \mathbb{F}^{*}_{q}$ such that $\mathbf{g}_{k-1}=a\cdot\mathbf{h}_{0}$.

Now we prove the theorem for the case with dimension $k=3$ and code length $n=6$. Since
the two vectors $\mathbf{g}_{0}, \mathbf{g}_{1}\in\mathrm{Span}\{\mathbf{h}_{0}, \mathbf{h}_{1}, \mathbf{h}_{2}\}$, we have
    \begin{itemize}
	\item there exists $ c_{0}, c_{1}, c_{2}\in\mathbb{F}_{q}$ such that
    \begin{small}
	\begin{displaymath}
		\frac{v^{2}_{i}\beta_{i}}{\gamma_{i}}=c_{0}+c_{1}\alpha_{i}+c_{2}f^{(1)}_{1}(\alpha_{i}, \beta_{i})
        \end{displaymath}
        \end{small}
	for any $i= 1, 2, \cdots, 6$.
	\item there exists $ e_{0}, e_{1}, e_{2}\in\mathbb{F}_{q}$ such that
    \begin{small}
	\begin{displaymath}
		\frac{v^{2}_{i}\beta_{i}}{\gamma_{i}}\alpha_{i}=e_{0}+e_{1}\alpha_{i}+e_{2}f^{(1)}_{1}(\alpha_{i}, \beta_{i})
        \end{displaymath}
        \end{small}
	for any $i=1, 2, \cdots, 6$.
\end{itemize}
Let $h_{c}(x, y)=c_{0}+c_{1}x+c_{2}f^{(1)}_{1}(x, y)$ and $h_{e}(x, y)=e_{0}+e_{1}x+e_{2}f^{(1)}_{1}(x, y)$. Then the polynomial $h_{c}(x, y)x-h_{e}(x,y)$ has $6$ different roots $\{(\alpha_{i}, \beta_{i}),\,i=1, 2, \cdots, 6\}$. Moreover, $h_{c}(x, y)x-h_{e}(x, y)\in\mathcal{L}(5O)$ if $c_2= 0$ and $h_{c}(x, y)x-h_{e}(x, y)\in\mathcal{L}(6O)\setminus \mathcal{L}(5O)$ if $c_2\neq 0$. For the former case, the polynomial $h_{c}(x, y)x-h_{e}(x, y)$ must be zero as before. For the latter case, from the choice of divisor $D$, we know that there exists $\alpha_{1},\alpha_{2},\alpha_{3}\in\mathbb{F}^{*}_{q}$ such that the polynomial $(x-\alpha_{i_{1}})(x-\alpha_{i_{2}})(x-\alpha_{i_{3}})$ vanishes on $\{(\alpha_{i}, \beta_{i}),\,i=1, 2, \cdots, 6\}$. Then the polynomial $h_{c}(x, y)x-h_{e}(x,y)-c_2\eta(x-\alpha_{i_{1}})(x-\alpha_{i_{2}})(x-\alpha_{i_{3}})$ has $6$ different roots $\{(\alpha_{i}, \beta_{i}),\,i=1, 2, \cdots, 6\}$. Note that $h_{c}(x, y)x-h_{e}(x,y)-c_2\eta(x-\alpha_{i_{1}})(x-\alpha_{i_{2}})(x-\alpha_{i_{3}})\in \mathcal{L}(5O)$. So
$$h_{c}(x, y)x-h_{e}(x,y)-c_2\eta(x-\alpha_{i_{1}})(x-\alpha_{i_{2}})(x-\alpha_{i_{3}})\equiv 0.$$
The coefficient of the term $xy$ in $h_{c}(x, y)x-h_{e}(x,y)-c_2\eta(x-\alpha_{i_{1}})(x-\alpha_{i_{2}})(x-\alpha_{i_{3}})$ is $c_2$, and hence $c_2=0$ which contradicts to the assumption $c_2\neq 0.$ In conclusion, we have 
$$h_{c}(x, y)x-h_{e}(x, y)\equiv 0.$$
So $e_{0}=e_{2}=c_{1}=c_{2}=0$ and $e_{1}=c_{0}$. The remaining proof is the same as the case $k\geq 5.$ 
\end{proof}


 \end{lemma}

\begin{theorem} 
Notations as above. For $n=2k$ ($k\geq 3$), we consider the TGECCs constructed over the first type of elliptic curves.
\begin{enumerate}
\item For odd $k$, TGECC $\mathcal{C}(D, kO,\ell, \eta, \mathbf{v})$ is self-dual if and only if 
\begin{enumerate}
\item $(\sum_{i=1}^{n}\frac{\gamma_{i}}{\beta_{i}}\alpha_{i}^{k+1})\eta+2\sum_{i=1}^{n}\gamma_{i}\alpha_{i}^{k-1}=0;$
\item there exists some $\lambda\in\mathbb{F}_{q}^{*}$ such that $v^{2}_{i}=\frac{\lambda\gamma_{i}}{\beta_{i}}$ for $ i=1, 2, \cdots, n$.
\end{enumerate}

\item For even $k$, TGECC $\mathcal{C}(D, kO,\ell, \eta, \mathbf{v})$ is self-dual if and only if 
\begin{enumerate}
\item $(\sum_{i=1}^{n}\gamma_{i}\beta_{i}\alpha_{i}^{k-2})\eta+2\sum_{i=1}^{n}\gamma_{i}\alpha_{i}^{k-1}=0;$
\item there exists some $\mu\in\mathbb{F}_{q}^{*}$ such that $v^{2}_{i}=\frac{\mu\gamma_{i}}{\beta_{i}}$ for $i=1, 2, \cdots, n$.
\end{enumerate}
\end{enumerate}
\label{59}
\vspace{-0.25 cm}
\end{theorem}
\begin{proof} We only give a proof for odd $k$ part. For the even $k$ part, the proof is similar.

The TGECC $\mathcal{C}(D, kO,\ell, \eta, \mathbf{v})$ is self-dual
if and only if $$\mathrm{Span}\{\mathbf{g}_{0}, \mathbf{g}_{1}, \cdots, \mathbf{g}_{k-1}\}=\mathrm{Span}\{\mathbf{h}_{0}, \mathbf{h}_{1}, \cdots, \mathbf{h}_{k-1}\}$$
 by Lemma~\ref{111} which is equivalent to the two following conditions:
    \begin{enumerate}
        \item $\mathrm{Span}\{\mathbf{g}_{0}, \mathbf{g}_{1}, \cdots, \mathbf{g}_{k-2}\}=\mathrm{Span}\{\mathbf{h}_{1}, \mathbf{h}_{1}, \cdots, \mathbf{h}_{k-1}\};$
        \item $\mathbf{g}_{k-1}=a\cdot\mathbf{h}_{0}$ for some $a\in \mathbb{F}^{*}_{q}.$
    \end{enumerate}
    
From the proof of Lemma~\ref{111} we have seen that the first condition $\mathrm{Span}\{\mathbf{g}_{0}, \mathbf{g}_{1}, \cdots, \mathbf{g}_{k-2}\}=\mathrm{Span}\{\mathbf{h}_{1}, \mathbf{h}_{1}, \cdots, \mathbf{h}_{k-1}\}$ is equivalent to that there exists some $\lambda\in\mathbb{F}_{q}^{*}$ such that $v^{2}_{i}=\frac{\lambda\gamma_{i}}{\beta_{i}}$ for $ i=1, 2, \cdots, n$.

The second condition $\mathbf{g}_{k-1}=a\cdot\mathbf{h}_{0}$ for some $a\in \mathbb{F}^{*}_{q}$ is equivalent to that the two ratios of coefficients equal:
\[
 \eta: 1=\left(-{\eta\sum^{n}_{i=1}\gamma_{i}\alpha^{k-1}_{i}}\right):\left({\sum^{n}_{i=1}\gamma_{i}\alpha^{k-1}_{i}+\eta\sum_{i=1}^{n}\frac{\gamma_{i}}{\beta_{i}}\alpha^{k+1}_{i}}\right)
\]
which is equivalent to $(\sum_{i=1}^{n}\frac{\gamma_{i}}{\beta_{i}}\alpha_{i}^{k+1})\eta+2\sum_{i=1}^{n}\gamma_{i}\alpha_{i}^{k-1}=0$ since $\eta\neq 0$.
\end{proof}

By the similar discussions as above, we have the results for the other two types of TGECCs.

\begin{corollary} 
Notations as above. For $n=2k$, we consider the TGECCs constructed over the second type of elliptic curves.
\begin{enumerate}
\item For odd $k$, TGECC $\mathcal{C}(D, kO,\ell, \eta, \mathbf{v})$ is self-dual if and only if $\sum_{i=1}^{n}\gamma_{i}\alpha_{i}^{k+1}=0$ and there exists some $\lambda\in\mathbb{F}_{q}^{*}$ such that $v^{2}_{i}=\lambda\gamma_{i}$ for $ i=1, 2, \cdots, n;$
  
\item For even $k$, TGECC $\mathcal{C}(D, kO,\ell, \eta, \mathbf{v})$ is self-dual if and only if $\sum_{i=1}^{n}\gamma_{i}\alpha_{i}^{k-2}\beta^{2}_{i}=0$
  and there exists some $\mu\in\mathbb{F}_{q}^{*}$ such that $v^{2}_{i}=\mu\gamma_{i}$ for $i=1, 2, \cdots, n$.
\end{enumerate}
\label{c2}
\end{corollary}

\begin{corollary} 
Notations as above. For $n=2k$, we consider the TGECCs constructed over the third type of elliptic curves.
\begin{enumerate}
\item For odd $k$, TGECC $\mathcal{C}(D, kO,\ell, \eta, \mathbf{v})$ is self-dual if and only if
$\sum_{i=1}^{n}\frac{\gamma_{i}}{a\alpha_{i}+b}\alpha^{k+1}_{i}=0$ and there exists some $\lambda\in\mathbb{F}_{q}^{*}$ such that $v^{2}_{i}=\frac{\lambda\gamma_{i}}{a\alpha_{i}+b}$ for $ i=1, 2, \cdots, n;$
  
\item For even $k$, TGECC $\mathcal{C}(D, kO,\ell, \eta, \mathbf{v})$ is self-dual if and only if $\sum_{i=1}^{n}\gamma_{i}(a\alpha_{i}+b)\alpha_{i}^{k-2}\beta_{i}=0$
  and there exists some $\mu\in\mathbb{F}_{q}^{*}$ such that $v^{2}_{i}=\frac{\mu\gamma_{i}}{a\alpha_{i}+b}$ for $i=1, 2, \cdots, n$.
\end{enumerate}
\label{c3}
\end{corollary}

\begin{corollary} \label{c4}
Consider the elliptic function field $F=\mathbb{F}_{q}(E)$   with $q$ a power of a odd prime $p$ and the defining equation is given by $E:\ y^{2}=f(x)$, where $f(x)$ is a square-free polynomial of degree $3$. If we have $ v_{i}=\sqrt{\frac{\lambda\gamma_{i}}{\beta_{i}}}$ for $i=1, 2, \cdots, 2k$
and $(\sum_{i=1}^{2k}v^{2}_{i}\alpha_{i}^{k+1})\eta+2\sum_{i=1}^{2k}v^{2}_{i}\beta_{i}\alpha_{i}^{k-1}=0$, then TGECC $\mathcal{C}(D, kO,\ell, \eta, \mathbf{v})$ over $\mathbb{F}_{q}$ is self-dual.
    \end{corollary}

\section{The Minimum Distances of TECCs}
\label{sec:6}

In this section, we will determine the possible minimum distance of TECCs by using the group structures.
	
	Let $E$ be an elliptic curve over $\mathbb{F}_{q}$ with one rational point $O$. The set $E(\mathbb{F}_q)$ of rational points on $E$ forms an abelian group with the identity element $O$. Moreover, the group of rational points $(E(\mathbb{F}_{q}), \oplus)\simeq\mathbb{Z}/n_{1}\times\mathbb{Z}/n_{2}$ for some $n_{1}\mid n_{2}$ and $n_{1}, n_{2}$ are two non-negative integers.
	
	For some $P\in E(\mathbb{F}_{q})$, we denote by $[m]P$ the addition of rational points on elliptic curves \textit{i.e.} $$[m]P=\underbrace{P\oplus\cdots\oplus P}_{m \ \text{times}}.$$ 
	Let $(E(\mathbb{F}_{q}),\oplus)$ be the abelian group. For any subset $D\subseteq E(\mathbb{F}_{q})\setminus\{O\}$, element $P\in E(\mathbb{F}_{q})$, and integer $1\leq k\leq \#D$, denote by $$N(k, P, D)=\#\left\lbrace T\subseteq D\ |\ P=\oplus_{Q\in T}Q\ \text{and}\ \#T=k\right\rbrace.$$
	
	The following Lemma states the relation between the existences of the rational functions and subset sum problems (SSPs).
	
	\begin{lemma}
		(see \cite{14})
		Assume $P_{1}, P_{2}, \cdots, P_{n}, P\in E(\mathbb{F}_{q})\setminus\{O\}$. If $m_{1}P_{1}+m_{2}P_{2}+\cdots+m_{n}P_{n}=P$, where $m_{i}, 1\leq i \leq n$, are positive integers, then there is a function vanishing at $P_{1}, P_{2}, \cdots, P_{n}$, with multiplicities $m_{1}, m_{2}, \cdots, m_{n}$, respectively, a pole at $P$ with multiplicity $1$ and a pole at $O$ with multiplicity $m_{1}+ m_{2}+\cdots+ m_{n}-1$.
		\label{22}
	\end{lemma}
	
	The relation between the minimum distance of ECCs and SSPs is listed in the following Lemma.
	\begin{lemma}(see \cite{8}) Notations as above. Let $E$ be a projective, non-singular and irreducible elliptic curve over $\mathbb{F}_{q}$, $D=\{P_{1}, P_{2}, \cdots, P_{n}\}$ a subset of $E(\mathbb{F}_{q})$ such that rational points (not necessarily distinct) $O\notin D$ and let $G=kO$ ($0<k<n$). Endow $E(\mathbb{F}_{q})$ a group structure with zero element $O$. The minimum distance of $\mathcal{C}_{\mathcal{L}}(D, kO)$ is $d=n-k+1$ if and only if $N(k, O, D)=0$ and the minimum distance $d=n-k$ if and only if $N(k, O, D)>0$.
	\end{lemma}

	Without loss of generality, we only give the proofs for the first type of TECCs with odd dimensions and denote it by $\mathcal{C}(D, kO,\ell, \eta)$.
	
	By Lemma \ref{22}, the subset sum problem $N(k+1, O, P)>0$ if and only if there exists at least one rational function $f(x, y)=\sum^{\frac{k+1}{2}}_{i=0}a_{i}x^{i}+\sum^{\frac{k-3}{2}}_{j=0}b_{j}x^{i}y\in\mathcal{L}((k+1)O)\in\mathcal{L}((k+1)O)$ vanishing at some $k+1$ rational places.
	For given rational points $P_{i_{1}}=(\alpha_{i_{1}}, \beta_{i_{1}}), P_{i_{2}}=(\alpha_{i_{2}}, \beta_{i_{2}}),\cdots, P_{i_{k+1}}=(\alpha_{i_{k+1}}, \beta_{i_{k+1}})$, we denote by $T=P_{i_{1}}+\cdots+P_{i_{k+1}}$. If there exists a function $f(x, y)\in\mathcal{L}((k+1)O)$ vanishing on them, then the coefficients can be uniquely determined and  such function $f(x,y)$ is unique {up to multiple}. 
	In this case, $f(x, y)\in S_{\ell}\subseteq\mathcal{L}((k+1)O)$ if and only if $\eta=\frac{a_{\frac{k+1}{2}}}{b_{\ell}}$ and such $\eta$ is denoted by $\eta=\eta(\ell, T)$. Note that  $\eta(\ell,T)$ do not exist for given $\ell$ if $N(k+1,O,D)=0$.
	
	\begin{theorem}
		Let $E$ be an elliptic curve over $\mathbb{F}_{q}$. Endow $E(\mathbb{F}_{q})$ a group structure with zero element as infinite place $O$. Denote by $d$ the minimum distance of  $\mathcal{C}(D, kO,\ell, \eta)$.
		\begin{enumerate}
			\item The minimum distance is $d=n-k-1$ if and only if 
			$\eta=\eta(\ell, T)$
			for some subset $T\subseteq \mathrm{Supp}(D)$ with $|T|=k+1$.
			\item  The minimum distance is $d=n-k$ if and only if 
			\begin{enumerate}
				\item $\eta\neq \eta(\ell, T)$ for any subset $T\subseteq \mathrm{Supp}(D)$ with $|T|=k+1$.
				\item $N(k, O, D)>0$ or
				\item if there is a place $P\in E(\mathbb{F}_{q})\setminus\mathrm{Supp}(D)$ or $P\in\mathrm{Supp}(\sum_{j=1}^{k}P_{i_{j}})$ such that $P+\sum^{k}_{j=1}P_{i_{j}}=O$ for some choices of $  \sum^{k}_{j=1}P_{i_{j}}\in\mathrm{Supp}(D)$, then $\eta=\eta(\ell, P+\sum^{k}_{j=1}P_{i_{j}})$.
			\end{enumerate}
			\item The minimum distance is $d=n-k+1$ if and only if 
			\begin{enumerate}
				\item $\eta\neq \eta(\ell, T)$ for any subset $T\subseteq \mathrm{Supp}(D)$ with $|T|=k+1$;
				\item $N(k, O, D)=0$ or there is no $P\in E(\mathbb{F}_{q})\setminus\mathrm{Supp}(D)$ or $P\in\mathrm{Supp}(\sum_{j=1}^{k}P_{i_{j}})$ such that $P+\sum^{k}_{j=1}P_{i_{j}}=O$ for some choices of $  \sum^{k}_{j=1}P_{i_{j}}\in\mathrm{Supp}(D)$ \item or $\eta\neq \eta(\ell, P+\sum^{k}_{j=1}P_{i_{j}})$ for any choice of $P\in E(\mathbb{F}_{q})\setminus\mathrm{Supp}(D)$ or $P\in\mathrm{Supp}(\sum_{j=1}^{k}P_{i_{j}})$.
			\end{enumerate}
		\end{enumerate}
		\label{10}
	\end{theorem}
	\begin{proof}
		By the relation $\mathcal{C}(D, kO, \ell,\eta)\subseteq\mathcal{C}_{\mathcal{L}}(D, (k+1)O)$, the minimum distance satisfies $d\geq d(\mathcal{C}_{\mathcal{L}}(D, (k+1)O))\geq n-(k+1)=n-k-1.$
		Combining the Singleton bound, the minimum distance of $\mathcal{C}(D, kO, \eta)$ only has three choices $\{n-k-1, n-k, n-k+1\}$. 
		
		1) If the minimum distance of $\mathcal{C}(D, kO,\ell, \eta)$ is $d=n-k-1$, by Lemma \ref{22}, there is a function $f\in S^{(1)}_{\ell}\subseteq\mathcal{L}((k+1)O)$ vanishing at $k+1$ points $T=P_{i_{1}}+ \cdots+ P_{i_{k+1}}\in \mathrm{Supp}(D)$. It follows $(f)=P_{i_{1}}+\cdots+P_{i_{k+1}}-(k+1)O$
		and then $\eta=\eta(\ell, T)$ for some $T\subseteq D$.

		2) If the minimum distance of $\mathcal{C}(D, kO,\ell, \eta)$ is $d=n-k$, by Lemma \ref{22}, there exists at least one rational function $g\in S^{(1)}_{\ell}\bigcap\mathcal{L}(kO)$ vanishing at $k$ points $T'=P_{j_{1}}+\cdots+ P_{j_{k}}\in \mathrm{Supp}(D)$ and 
		$(f')=P_{i_{1}}+\cdots+P_{i_{k}}-kO$, then we have $N(k, O, D)>0$. We also need $\eta\neq\eta(\ell,T)$
		for any subset $T\subseteq\mathrm{Supp}(D)$. If not, there exists at least one rational function whose index of pole $O$ exceeds $k$, which is a contradiction. We also need to classify the cases in which there is a rational place $P\in E(\mathbb{F}_{q})\setminus\mathrm{Supp}(D)$ or $P\in\mathrm{Supp}(\sum^{k}_{j=1}P_{i_{j}})$ such that $P+\sum^{k}_{j=1}P_{i_{j}}=O$ for some choices of $  \sum^{k}_{j=1}P_{i_{j}}\in\mathrm{Supp}(D)$, then we need $\eta=\eta(\ell, P+\sum^{k}_{j=1}P_{i_{j}})$ to ensure the existence of function in $S_{\ell}\subseteq\mathcal{L}((k+1)O)$ vanishing at $k$ places in $\mathrm{Supp}(D)$.
		
		3) If the minimum distance of $\mathcal{C}(D, kO,\ell, \eta)$ is $d=n-k+1$, by three choices of the minimum distance of TECC $\mathcal{C}(D, kO,\ell, \eta)$, it follows that if $\eta\neq\eta(\ell, T)$ for any subset $T\subseteq \mathrm{Supp}(D),$ then the minimum distance $d$ of $\mathcal{C}(D, kO,\ell, \eta)$ cannot satisfy $d=n-k-1$. 
		Together with $N(k, O, D)=0$ and no $P\in E(\mathbb{F}_{q})\setminus\mathrm{Supp}(D)$ or $P\in\mathrm{Supp}(\sum^{k}_{j=1}P_{i_{j}})$ such that $P+\sum^{k}_{j=1}P_{i_{j}}=O$ for some choices of $  \sum^{k}_{j=1}P_{i_{j}}\in\mathrm{Supp}(D)$ (resp. $\eta\neq \eta(\ell, P+\sum^{k}_{j=1}P_{i_{j}})$ for any choice of $P\in E(\mathbb{F}_{q})\setminus\mathrm{Supp}(D)$ or $P\in\mathrm{Supp}(\sum_{j=1}^{k}P_{i_{j}})$), then $\mathcal{C}(D, kO,\ell, \eta)$ is MDS. 
		
	\end{proof}
	
	\begin{remark}
		If we consider long TECCs, \textit{is, } $|D|\geq q+1$, then there are only two choices for TECCs, that \textit{is, }\ $\{n-k-1,n-k\}$ by the MDS conjecture in \cite{20}. If $\eta\neq \eta(\ell,T)$, then the minimum distance is $d=n-k$ in this case, which also means $\mathcal{C}(D,kO,\ell,\eta)$ is NMDS. In the last section, we shall prove the non-equivalence between TECCs and ECCs. It also means we find a new construction of NMDS codes which is not monomially equivalent to ECCs.
	\end{remark}
	
	\begin{corollary} 
		The TGECC $\mathcal{C}(D, kO,\frac{k-3}{2}, \eta, \mathbf{v})$ is self-dual MDS if and only if
		\begin{enumerate}
			\item $N(k, O, D)=0$ and there is no $P\in E(\mathbb{F}_{q})\setminus\mathrm{Supp}(D)$ or $P\in\mathrm{Supp}(\sum_{j=1}^{k}P_{i_{j}})$ such that $P+\sum^{k}_{j=1}P_{i_{j}}=O$ for some choices of $  \sum^{k}_{j=1}P_{i_{j}}\in\mathrm{Supp}(D)$.
			\item There exists some $\lambda\in\mathbb{F}_{q}^{*}$ such that $v^{2}_{i}=\frac{\lambda\gamma_{i}}{\beta_{i}}$ for $ i=1, 2, \cdots, n;$
			\item
			$\eta=-\frac{2\sum_{i=1}^{n}\gamma_{i}\alpha_{i}^{s-1}}{\sum_{i=1}^{n}\frac{\gamma_{i}}{\beta_{i}}\alpha_{i}^{s-1}}\neq\eta\left(\small{\frac{k-3}{2}},T\right)$
			for any subset $T\subseteq Supp(D)$ with $|T|=k+1$.
		\end{enumerate}
	\end{corollary}

    Now we provide an explicit example of self-dual MDS code to better illustrate the corollary above.
    \begin{example}
        Let $F=\mathbb{F}_{11}(E)$ be an elliptic function field defined by $E:y^{2}=x^{3}+x+1.$ It can be verified that $E(\mathbb{F}_{11})\simeq\mathbb{Z}/14\mathbb{Z}$ with generator $R=(0,1)$. Take $D=\{R,3R,5R,9R,11R,13R\}$ with 
        \begin{displaymath}
            3R=(8,2), 5R=(4,6), 9R=(4,5), 11R=(8,9), 13R=(1,6).
        \end{displaymath}
        Take the defining set $S_{0}=\{a_{0}+a_{1}x+b_{0}(y+\eta x^{2})\ \big|\ a_{i},b_{j}\in\mathbb{F}_{11}\}$. Note that $N(3,O,D)=0$ and $N(4,O,D)>0$, by Theorem \ref{10}, then $\mathcal{C}(D,3O,0,\eta)$ is an MDS code if $\eta\neq\{4,5,6,10\}$.
    \end{example}

\begin{remark}
1) If we consider long TECCs, \textit{is, } $|D|\geq q+1$, then there are only two choices for TECCs, that \textit{is, }\ $\{n-k-1,n-k\}$ by the MDS conjecture in \cite{20}.

2) The above discussions can also be applied to determine the minimum distance of the TGRS code $\mathcal{C}(D,kO,\ell,\eta)$ with the valuation divisor $D=P_{1}+\cdots+P_{n}$ and each $P_{i}$ corresponding to $x-\alpha_{i},\alpha_{i}\in\mathbb{F}_{q}$. The defining set is given by $V_{\ell}=\{\sum_{i=0}^{k-1}a_{i}x^{i}+a_{\ell}\eta x^{k}, a_{i}\in\mathbb{F}_{q},\eta\in\mathbb{F}_{q}^{*}\}$. According to Goppa's bound, the minimum distance of $\mathcal{C}(D,kO,\ell,\eta)$ only have two choices \textit{i.e.} $\{n-k,n-k+1\}$. If the minimum distance is $n-k$, then there exists at least one function $f\in V_{\ell}\subseteq\mathcal{L}((k+1)O)$ vanishing at $k$ rational places. Then we can assume $(f)=P_{i_{1}}+P_{i_{2}}+\cdots+P_{i_{k}}\subseteq D$ with each $P_{i_{j}}$ corresponding to $x-\alpha_{i_{j}}, \alpha_{i_{j}}\in\mathbb{F}_{q}$ and it is equivalent to $f=\prod^{k}_{j=1}(x-\alpha_{i_{j}})$. To ensure $f\in V_{\ell}$, we have $\frac{1}{\eta}x^{\ell}+x^{k}=(-1)^{k-\ell}\sigma_{k-\ell}(P_{i_{1}},\cdots,P_{i_{k}})x^{\ell}+x^{k}$ which induces $\eta=(-1)^{k-\ell}/\sigma_{k-\ell}(P_{i_{1}},\cdots,P_{i_{k}})$ where $\sigma_{k-\ell}(P_{i_{1}},\cdots,P_{i_{k}})$ is the elementary polynomial with variables $\alpha_{i_{1}},\cdots,\alpha_{i_{k}}\in\mathbb{F}_{q}$. If the TGRS code $\mathcal{C}(D,kO,\ell,\eta)$ is MDS, then $\eta\neq (-1)^{k-\ell}/\sigma_{k-\ell}(P_{i_{1}},\cdots,P_{i_{k}})$ for any subset $\{P_{i_{1}},\cdots,P_{i_{k}}\}\subseteq D$, which coincides with the results given in \cite{5,18} \textit{etc.}    
 \end{remark}

\begin{example} Let $F=\mathbb{F}_{4}(E)$ be an elliptic function field with the defining equation $E:\ y^{2}+y=x^{3}$ and $\mathbb{F}_{4}=\mathbb{F}_{2}(\alpha)$ where $\alpha^{2}+\alpha+1=0$. It is easy to verify that the infinity $O$ and $ P_{1}=(1, \alpha), P_{2}=(1,\alpha+1),
P_{3}=(\alpha, \alpha),P_{4}=(\alpha, \alpha+1),P_{5}=(\alpha+1, \alpha),
P_{6}=(\alpha+1, \alpha+1), P_{7}=(0,1), P_{8}=(0, 0)$ are all the rational points on the elliptic curve $E$. 

First, we take $D=\sum_{i=1}^{6}P_{i}$ and any $\mathbf{v}=(v_{1}, v_{2}, \cdots, v_{6})\in(\mathbb{F}^{*}_{4})^{6}$, and construct a $[6, 3, 3]$ GECC $\mathcal{C}_{\mathcal{L}}(D, 3O, \mathbf{v})$ whose generator matrix can be given by
\begin{small}
\begin{displaymath}
	\begin{pmatrix}
		v_{1}& v_{2}& v_{3}&v_{4}&v_{5}&v_{6}\\
		v_{1}& v_{2}& v_{3}\alpha&v_{4}\alpha&v_{5}(\alpha+1)&v_{6}(\alpha+1)\\
		v_{1}\alpha&v_{2}(\alpha+1)&v_{3}\alpha&v_{4}(\alpha+1)&v_{5}\alpha&v_{6}(\alpha+1)\\
	\end{pmatrix}.
\end{displaymath}
\end{small}
Then by choosing the following defining set
\begin{small}
\begin{displaymath}
 S_{0}=\left\lbrace a_{0}+a_{1}x+b_{0}(y+\eta x^{2})\ \bigg|\ a_1, a_2, b_0\in\mathbb{F}_{q},\ \eta\in\mathbb{F}^{*}_{q}\right\rbrace 
\end{displaymath}
\end{small}
we construct the TGECC $\mathcal{C}(D, 3O,0, \eta, \mathbf{v})$ which has a generator matrix
\begin{small}
\begin{displaymath}
	\begin{pmatrix}
		v_{1}& v_{2}& v_{3}&v_{4}&v_{5}&v_{6}\\
		v_{1}& v_{2}& v_{3}\alpha&v_{4}\alpha&v_{5}(\alpha+1)&v_{6}(\alpha+1)\\
		v_{1}(\alpha+\eta)&v_{2}(\alpha+1+\eta)&v_{3}(\alpha+\eta(\alpha+1))&v_{4}(\alpha+1+\eta(\alpha+1))&v_{5}(\alpha+\eta \alpha)&v_{6}(\alpha+1+\eta\alpha)\\
	\end{pmatrix}
\end{displaymath}
\end{small}
and a parity-check matrix
\begin{small}
\begin{displaymath}
	\begin{pmatrix}
		\frac{1}{v_{1}}& \frac{1}{v_{2}}& \frac{1}{v_{3}(\alpha+1)}&\frac{1}{v_{4}(\alpha+1)}&\frac{1}{v_{5}\alpha}&\frac{1}{v_{6}\alpha}\\
		\frac{1}{v_{1}}& \frac{1}{v_{2}}& \frac{\alpha}{v_{3}(\alpha+1)}&\frac{\alpha}{v_{4}(\alpha+1)}&\frac{\alpha+1}{v_{5}\alpha}&\frac{\alpha+1}{v_{6}\alpha}\\
		\frac{\alpha+\eta}{v_{1}}&\frac{\alpha+1+\eta}{v_{2}}&\frac{\alpha+\eta(\alpha+1)}{v_{3}(\alpha+1)}&\frac{\alpha+1+\eta(\alpha+1)}{v_{4}(\alpha+1)}&\frac{\alpha+\eta\alpha}{v_{5}\alpha}&\frac{\alpha+1+\eta\alpha}{v_{6}\alpha}\\
	\end{pmatrix}
    =			\begin{pmatrix}\frac{1}{v_{1}}& \frac{1}{v_{2}}& \frac{\alpha}{v_{3}}&\frac{\alpha}{v_{4}}&\frac{\alpha+1}{v_{5}}&\frac{\alpha+1}{v_{6}}\\
		\frac{1}{v_{1}}& \frac{1}{v_{2}}& \frac{\alpha+1}{v_{3}}&\frac{\alpha+1}{v_{4}}&\frac{\alpha}{v_{5}}&\frac{\alpha}{v_{6}}\\
		\frac{\alpha+\eta}{v_{1}}&\frac{\alpha+1+\eta}{v_{2}}&\frac{\alpha+1+\eta}{v_{3}}&\frac{1+\eta}{v_{4}}&\frac{1+\eta}{v_{5}}&\frac{\alpha+\eta}{v_{6}}\\
	\end{pmatrix}
\end{displaymath}
\end{small}
by the relation $\alpha^{2}=\alpha+1$.

Take $\mathbf{v}=(\lambda,\lambda,\frac{\lambda}{\alpha},\frac{\lambda}{\alpha},\frac{\lambda}{\alpha+1},\frac{\lambda}{\alpha+1})$ for any $\lambda\in\mathbb{F}^{*}_{4}$. By self-duality condition,  the  TGECC $\mathcal{C}(D, 3O,0, \eta, \mathbf{v})$ is self-dual. Indeed, we write down the transformation matrix between the generator matrix and parity-check matrix of TGECC $\mathcal{C}(D, 3O,0, \eta, \mathbf{v})$ as below:
\begin{small}
\begin{displaymath}
	\begin{pmatrix}
		\lambda& \lambda& \lambda(\alpha+1)&\lambda(\alpha+1)&\lambda\alpha&\lambda\alpha\\
		\lambda& \lambda& \lambda&\lambda&\lambda&\lambda\\
		\lambda(\alpha+\eta)&\lambda(\alpha+1+\eta)&\lambda(1+\eta\alpha)&\lambda(\alpha+\eta\alpha)&\lambda(\alpha+1+\eta (\alpha+1))&\lambda(1+\eta(\alpha+1))\\
        \end{pmatrix}\\
        \end{displaymath}
        \end{small}
        \begin{small}
        \begin{displaymath}
=\begin{pmatrix}
		\lambda^{2}& 0& 0\\
		0&\lambda^{2}&0\\
        0&0&\lambda^{2}
	\end{pmatrix}\begin{pmatrix}
		\frac{1}{\lambda}& \frac{1}{\lambda}& \frac{\alpha+1}{\lambda}&\frac{\alpha+1}{\lambda}&\frac{\alpha}{\lambda}&\frac{\alpha}{\lambda}\\
		\frac{1}{\lambda}& \frac{1}{\lambda}& \frac{1}{\lambda}&\frac{1}{\lambda}&\frac{1}{\lambda}&\frac{1}{\lambda}\\
		\frac{\alpha+\eta}{\lambda}&\frac{\alpha+1+\eta}{\lambda}&\frac{1+\eta\alpha}{\lambda}&\frac{\alpha+\eta\alpha}{\lambda}&\frac{\alpha+1+\eta(\alpha+1)}{\lambda}&\frac{1+\eta(\alpha+1)}{\lambda}\\
	\end{pmatrix}.
    \end{displaymath}
\end{small}

Now we compute the minimum distance of  the  TGECC $\mathcal{C}(D, 3O,0, \eta, \mathbf{v}):$
\[
d(\mathcal{C}(D, 3O,0, \eta, \mathbf{v}))=d(\mathcal{C}(D, 3O,0, \eta)).
\]
We have the following group isomorphism
$$E(\mathbb{F}_{4})\simeq\mathbb{Z}/3\times\mathbb{Z}/3.$$
The isomorphism is shown in Table \ref{667}.
\begin{table}[h]
    \centering
    \begin{tabular}{|l|l|}
    \hline
      $\mathbb{Z}/3\times\mathbb{Z}/3$& Corresponding Rational Point \\\hline
      $(0, 0)$& $[0]P_{1}\oplus [0]P_{3}=O$\\ \hline 
  $(1, 0)$& $P_{1}$   \\ \hline     
  $(0, 1)$& $ P_{3}$   \\ \hline 
          $(1, 1)$& $P_{1}\oplus P_{3}=P_6$\\ \hline         
          $(1, 2)$& $ P_{1}\oplus [2]P_{3}=P_7$\\ \hline      
  $(2, 0)$& $[2]P_{1}=P_2$\\ \hline 
  $(2, 1)$& $[2]P_{1}\oplus P_{3}=P_8$\\ \hline 
  $(2, 2)$& $[2]P_{1}\oplus [2]P_{3}=P_5$\\ \hline 
  $(0, 2)$& $ [2]P_{3}=P_4$\\ \hline
        \end{tabular}
    \caption{The group structure of $E(\mathbb{F}_{4})$}
    \label{tab:my_label}
    \label{667}
\end{table}

To compute the value of $N(4, O, D)$, we first translate the problem to the group 
$\mathbb{Z}/3\times\mathbb{Z}/3$ via the isomorphism and it is not difficult to give the solutions:
\[
\{\{(1,0), (2,0), (0,1), (0,2)\},\{(1,0), (2,0), (2,2), (1,1)\},\{(0,1), (0,2), (2,2), (1,1)\}\}.
\]
So we obtain the solutions of $N(4, O, D)$ and the corresponding rational functions shown in Table~\ref{666}.
\begin{table}[h]
   \centering
  \begin{tabular}{|c|c|c|c|c|c|}
   \hline
No.&$P$&$Q$&$R$&$S$& The Corresponding Rational Function \\\hline
$1$&$P_1=(1, \alpha)$&$P_2=(1, \alpha+1)$&$P_3=(\alpha, \alpha)$&$P_4=(\alpha, \alpha+1)$& $f(x, y)=\alpha+(\alpha+1)x+x^{2}$\\\hline
$2$&$P_3=(1, \alpha)$&$P_4=(1, \alpha+1)$&$P_5=(\alpha+1, \alpha)$&$P_6=(\alpha+1, \alpha+1)$& $f(x, y)=\alpha+1+\alpha x+x^{2}$\\\hline
$3$&$P_1=(\alpha, \alpha)$&$P_2=(\alpha, \alpha+1)$&$P_5=(\alpha+1, \alpha)$&$P_6=(\alpha+1, \alpha+1)$& $f(x, y)=1+x+x^{2}$\\\hline
  \end{tabular}
    \caption{Solutions of $N(4, O, D)$ }
    \label{666}
    \end{table}

Since the rational functions do not belong to the defining set $S_0$, the minimum distance satisfies $d(\mathcal{C}(D, 3O,0, \eta))\geq 3.$ Similarly, we have $N(3, O, D)=2$ and the solutions of $P\oplus Q\oplus R=O$ for $P, Q, R\in \mathrm{Supp}(D)$ are given by
\[
\{\{P_1, P_3, P_5\}, \{P_2, P_4, P_6\}\}.
\]
The corresponding rational functions vanishing on the solution points are given in Table~\ref{668}. 
\begin{table}[h] \label{668}
\centering
    \begin{tabular}{|c|c|c|c|c|}
    \hline
No.&$P$&$Q$&$R$& The Corresponding Rational Function \\\hline
$1$&$P_1=(1, \alpha)$&$P_3=(\alpha, \alpha)$&$P_5=(\alpha+1, \alpha)$& $f(x, y)=y+\alpha$\\\hline
$2$&$P_2=(1, \alpha+1)$&$P_4=(\alpha, \alpha+1)$&$P_6=(\alpha+1, \alpha+1)$& $f(x, y)=y+\alpha+1$\\\hline
   \end{tabular}
    \caption{Solutions of $N(3, O, D)$ }
    \end{table}

Note that the two rational functions in Table \ref{668} do not belong to the defining set $S_{0}$ which means the minimum distance $d(\mathcal{C}(D, 3O,0, \eta))\geq 4.$ Therefore, the TECC $\mathcal{C}(D, 3O,0, \eta)$ is a $[6,3,4]$ MDS code and the TGECC $\mathcal{C}(D, 3O,0, \eta, \mathbf{v})$ is  a $[6,3,4]$ MDS self-dual code for the above chosen $\mathbf{v}$.\qed
\vspace{-0.25 cm}
\end{example}

\begin{example} Let $F=\mathbb{F}_{5}(E)$ be an elliptic function field  with defining equation $E:\ y^{2}=x^{3}+x+1$. Let $F=\mathbb{F}_{5}(E)$ be an elliptic function field  with the defining equation:
		\begin{displaymath}
			E:\ y^{2}=x^{3}+x+1.
		\end{displaymath}
		The elliptic curve $E$ has 9 rational points: the infinity place $O$ and
		\begin{small}
			\begin{displaymath}
				\begin{matrix}
					P_{1}=(0, 1)&P_{2}=(0, -1)& P_{3}=(2,1 )&P_{4}=(2,-1)\\
					P_{5}=(-2, 1)& P_{6}=(-2, -1)& P_{7}=(-1, 2)& P_{8}=(-1, -2).
				\end{matrix}
			\end{displaymath}
		\end{small}
		Choose $D=P_1+P_2+\cdots+P_8$. 
		
		Let $t=x(x-2)(x-3)(x-4)$ be a local uniformizer and consider the differential $\omega=dx/t$. 
		Then we have the residues:
		$\mathrm{res}_{P_{1}}(\omega)=\mathrm{res}_{P_{2}}(\omega)=-\frac{1}{24}$, $\mathrm{res}_{P_{3}}(\omega)=\mathrm{res}_{P_{4}}(\omega)=\frac{1}{4}$, $\mathrm{res}_{P_{5}}(\omega)=\mathrm{res}_{P_{6}}(\omega)=-\frac{1}{3}$ and $\mathrm{res}_{P_{7}}(\omega)=\mathrm{res}_{P_{8}}(\omega)=\frac{1}{8}$.

		Take the defining set 
		\begin{small}
			\begin{displaymath}
				S_{0}=\left\lbrace a_{0}+a_{1}x+b_{0}\left( y+\eta x^{2}\right) \ \bigg|\ a_{i},b_{j}\in\mathbb{F}_{5}\right\rbrace. 
			\end{displaymath}
		\end{small}
		TECC $\mathcal{C}(D, 3O, 0, \eta)$ has a generator matrix
		\begin{small}
			\begin{displaymath}
				\begin{pmatrix}
					1& 1& 1&1&1&1&1&1\\
					0& 0& 2&2&3&3&4&4\\
					1& 4& 1+4\eta&4+4\eta&1+4\eta&4+4\eta& 2+\eta&3+\eta\\
				\end{pmatrix}
			\end{displaymath}
		\end{small}
		and a parity-check matrix
		\begin{small}
			\begin{displaymath}
				\begin{pmatrix}
					1& 4& 1&4&3&2&1&4\\
					0& 0& 2&3&4&1&4&1\\
					0& 0& 4&1&2&3&1&4\\
					1&1&4&4&3&3&2&2\\
					0& 0& 3+3\eta&3+2\eta&4+4\eta&4+\eta&3+\eta&3+4\eta\\
				\end{pmatrix}.
			\end{displaymath}
		\end{small}
		
		Now we compute the minimum distance of $\mathcal{C}(D, 3O, 0, \eta)$. First, it can be checked that the group $E(\mathbb{F}_{5})\simeq\mathbb{Z}/9$. Take $P_{1}$ as the primitive element, then $P_{7}=[2]P_{1}$, $P_{3}=[3]P_{1}$, $P_{6}=[4]P_{1}$, $P_{5}=[5]P_{1}$, $P_{4}=[6]P_{1}$, $P_{8}=[7]P_{1}$, $P_{2}=[8]P_{1}$, $O=[9]P_{1}$.
		
		We first compute the value of $N(4, O, D)$ which is equivalent to enumerate the possible combinations of $P, Q, R, S\in \mathrm{Supp}(D)$ such that $P\oplus Q\oplus R\oplus S=O$. 
		
		Let $P=[a]P_{1}, Q=[b]P_{1}, R=[c]P_{1}, S=[d]P_{1}\in \mathrm{Supp}(D)= E(\mathbb{F}_{5})\setminus\{O\}$ for some pairwise distinct integers $a, b, c, d\in\{1, \cdots, 8\}$, then the above counting problem is equivalent to counting combinations such that $a+b+c+d\equiv 0\ (\text{mod}\ 9)$. Note that $a+b+c+d\in\{10, \cdots, 26\}$. The only possibility is $a+b+c+d=18$. Then we have the following solutions and the corresponding rational functions in Table~\ref{899}:
		\begin{table}[h]
			\centering
			\begin{tabular}{|c|c|c|}
				\hline
				No.&$\{a, b, c, d\}$& The Corresponding Rational Function \\\hline
				$1$&$\{1,2,7,8\}$& $f(x, y)=x+x^{2}$\\\hline
				$2$&$\{1,3,6,8\}$& $f(x, y)=3x+x^{2}$\\\hline
				$3$&$\{1,4,5,8\}$& $f(x, y)=2x+x^{2}$\\\hline
				$4$&$\{1,4,6,7\}$& $f(x, y)=4+y+3x^{2}$\\\hline
				$5$&$\{2,3,5,8\}$& $f(x, y)=1+y+2x^{2}$\\\hline
				$6$&$\{2,3, 6,7\}$& $f(x, y)=3+4x+x^{2}$\\\hline
				$7$&$\{2,4,5,7\}$& $f(x, y)=2+3x+x^{2}$\\\hline
				$8$&$\{3,4,5,6\}$& $f(x, y)=1+x^{2}$\\\hline
			\end{tabular}
			\caption{Solutions of $N(4,O,D)$}
			\vspace{-1 cm}
			\label{899}
		\end{table}
		
		Then we have $N(4, O, D)=8$. Note that only the fourth and fifth rational functions are in the defining set $S_{0}$. Then the minimum distance satisfies $d(\mathcal{C}(D, 3O, 0, \eta))=4$ if and only if $\eta\in \{2, 3\}$. 
		
		On the other hand, by the MDS conjecture of ECCs in \cite{11,12}, we know that $\mathcal{C}(D, 3O,0, \eta)$ cannot be MDS. Therefore TECC $\mathcal{C}(D, 3O,0, \eta)$ is AMDS for  $\eta\in\{1,4\}$.\qed
	\end{example}

\section{Non-equivalence Results }
	\label{sec:6}
	In \cite{4}, we know that any MDS code of length $n$ and dimension $k$ is equivalent to an RS code for $k<3$ and $n-k<3$. In this subsection, we present some non-equivalence results of TECCs for $4\leq k\leq \frac{n-4}{2}$ and $\frac{n+4}{2}\leq k\leq  n-4$. First, the following proposition calculates the dimensions of Schur squares of the classical ECCs $\mathcal{C}_{\mathcal{L}}(D, kO)$ for $4\leq k\leq \frac{n-4}{2}$.
	
	\begin{proposition}[see \cite{3}]
		The dimension of the Schur square of ECC $C_{\mathcal{L}}(D, kO)$ is $2k$ if $4\leq k\leq \frac{n-4}{2}$.
		\label{7}
	\end{proposition}
	Consider the TGRS code $C((k-1)O,\ell,\eta)$ with defining set $R_{\ell}=\{\sum^{k-1}_{i=0}a_{i}x^{i}+a_{\ell}\eta x^{k}, a_{i}\in\mathbb{F}_{q},\eta\in\mathbb{F}_{q}^{*}\}$. It has been verified that the dimension of Schur square of $C((k-1)O,\ell,\eta)$ is lower bounded by $2k$; see \cite{4}. By checking the maximal order of the polynomial in $R_{\ell}$, the dimension of Schur square is upper bounded by $2k+1$ if $\ell\neq k-1$.

	Based on the results above, we can calculate the dimensions of the Schur squares of TECCs and the following theorem illustrates the monomial non-equivalence between TECCs and ECCs, GRS and TGRS codes.
	\begin{theorem}
		\begin{enumerate}
			\item For $4\leq k\leq \frac{n-4}{2}$, the dimension of the Schur square of $C(D, kO,\ell, \eta)$ is at least $2k+1$. Moreover, the dimension of the Schur square of $C(D, kO,\ell, \eta)$ equals to $2k+1$ if and only if $\ell=\frac{k-3}{2}$ for odd $k$ or $\ell=\frac{k}{2}$ for even $k$.
			\item For $4\leq n-k\leq\frac{n-4}{2}$, the dimension of Schur square of the dual $C(D, kO,\ell, \eta)^{\perp}$ is at least $2n-2k+1$. Moreover, the dimension of the Schur square of $C(D, kO,\ell, \eta)^{\perp}$ equals to $2n-2k+1$ if and only if $\ell=\frac{k-3}{2}$ for odd $k$ or $\ell=\frac{k}{2}$ for even $k$.
		\end{enumerate}
		\label{30} 
	\end{theorem}
	Together with Propositions~\ref{3} and~\ref{7}, we have the following corollary.
	\begin{corollary}
		\begin{enumerate}
			\item TECC $C(D, kO,\ell, \eta)$ is not monomially equivalent to any $k-$dimensional GRS code for $4\leq k\leq \frac{n-4}{2}$ or $\frac{n+4}{2}\leq k\leq n-4;$
			\item TECC $C(D, kO,\ell, \eta)$ is not monomially equivalent to any $k-$dimensional ECCs $\mathcal{C}_{\mathcal{L}}(D,kO)$ for $4\leq k\leq \frac{n-4}{2}$ or $\frac{n+4}{2}\leq k\leq n-4$.
			\item TECC $C(D, kO,\ell, \eta)$ with $0\leq\ell< \frac{k-3}{2}$ or $0\leq\ell<\frac{k}{2}$ is not monomially equivalent to any TGRS code $C((k-1)O,\phi,\eta)$ for $4\leq k\leq \frac{n-4}{2}$ or $\frac{n+4}{2}\leq k\leq n-4$.
		\end{enumerate}
	\end{corollary}
	\begin{proof}
		For $4\leq k\leq \frac{n-4}{2}$, it is straightforward from Theorem~\ref{30} 1), Propositions~\ref{3} and~\ref{7} by comparing the dimensions of the Schur squares.  For $\frac{n+4}{2}\leq k\leq n-4$, we consider dual codes since the dual codes of GRS codes (and ECCs) are still GRS codes (and generalized ECCs, see \cite{10}).
	\end{proof}

	To prove Theorem \ref{30}, we need the following Propositions and Lemmas, which are generalizations of operations for TGRS codes by Beelen \textit{et al}. in \cite{4}. 
	
	Without loss of generality, we only give the proof of the TECCs $\mathcal{C}(D, kO, \ell, \eta)$ with odd dimensions and simply denote by $S_{\ell}=S^{(1)}_{\ell}$ in the following discussions.
	
	Consider the defining set $S_{\ell}\subseteq\mathcal{L}((k+1)O)$ and $D=\sum_{i=1}^{n}P_{i}$ given as before. Denote by two sets $D(S_{\ell}):=\{-v_{O}(f\cdot g): f, g\in S_{\ell}, v_{O}(f\cdot g)>-n    \}$ and $\overline{D}(S_{\ell}, D):=\{-v_{O}(\overline{f\cdot g}): f, g\in S_{\ell}    \}$, where $\overline{f}:=(f\mod h)$ for any $f\in\mathcal{L}(nO)$  and $h$ is the rational function with all the places in $\mathrm{Supp}(D)$ as zeros. 
	
	\begin{remark}
		\begin{enumerate}
			\item Note that the existence of $f, g, h$ rely on the SSPs in $\mathrm{Supp}(D)$.
			\item In the setting of $4\leq k\leq \frac{n-4}{2}$, the cardinality of the two sets $\overline{D}(S_{\ell}, D)$ and $D(S_{\ell})$ are the same and we have
			\begin{displaymath}
				ev_{D}\left(\left\langle fg: f, g\in S_{\ell}\right\rangle\right)= ev_{D}\left(\left\langle\overline{fg}: f, g\in S_{\ell}\right\rangle\right).
			\end{displaymath}
		\end{enumerate}
	\end{remark}
	
	\begin{lemma}
		\begin{enumerate}
			\item $\mathcal{C}(D, kO, \ell, \eta)^{\star 2}=ev_{D}\left(\left\langle fg: f, g\in S_{\ell}\right\rangle\right);$
			\item $\dim\left(\mathcal{C}(D, kO, \ell, \eta)^{\star 2}\right)\geq \#\overline{D}(S_{\ell}, D)\geq \#D(S_{\ell}).    $ 
		\end{enumerate}
		\label{113}
		\begin{proof} The first part of the statement follows directly from the definition of Schur square, \textit{i.e.}, $f(P)\cdot g(P)=(f\cdot g)(P)$ for $f, g\in S_{\ell}$ and any $P\in \mathrm{Supp}(D)$. Note that 
			the evaluation $ev_{D}(\cdot)$ is a bijection between $\mathcal{L}(nO)$ and $\mathbb{F}^{n}_{q}$. Then $\dim(C(D, O, \eta, \ell)^{\star 2})$ is greater than or equal to the cardinality $\#\overline{D}(S_{\ell}, D)$ of the set $\left\langle \overline{f\cdot g}: f, g\in S_{\ell}\right\rangle$. Note that $D(S_{\ell})\subseteq\overline{D}(S_{\ell}, D)$, then we have the second inequality.
		\end{proof}
	\end{lemma}

	\begin{proposition}
		\vspace{-0.25 cm}
		Notations as above. Denote by $T_{\ell}=\{0, 2, 3, \cdots, k, k+1\}\setminus\{2\ell+3\}$ or $T_{\ell}=\{0, 2, 3, \cdots, k, k+1\}\setminus\{2\ell\}$. Then the dimension of the Schur square satisfies
		\begin{displaymath}
			\dim\left(\mathcal{C}(D, kO, \ell, \eta)^{\star 2}\right)\geq\#\{d_{1}+d_{2}: d_{1},d_{2}\in T_{\ell}, d_{1}+d_{2}<n\}.
		\end{displaymath}
		\begin{proof} 
			For simplicity of discussions, we only prove for the first case. Denote by $f_{1}=1,\cdots,f_{2\ell+3}=x^{\ell}y+\eta x^{\frac{k+1}{2}},\cdots, f_{k}=x^{\frac{k-3}{2}}y$. By checking their indices of poles, it is obvious that they form a basis of $S_{\ell}$ and $T_{\ell}=\{0,2,3,\cdots,k+1\}\setminus\{2\ell+3\}=\{-v_{O}(f_{1}),\cdots,-v_{O}(f_{k})\}$. Then the proposition directly from Lemma \ref{113}.
		\end{proof}
		\label{114}
	\end{proposition}
	
	\begin{proof}[Proof of Theorem \ref{30}]
		\begin{enumerate}
			\item By Proposition \ref{114}, we have $T_{\ell}=\{0, 2, \cdots, k,k+1\}\setminus\{2\ell+3\}$ and then we have the following discussions for the set
			$\{d_{1}+d_{2}: d_{1},d_{2}\in T_{\ell}, d_{1}+d_{2}<n\}$. 
			\begin{enumerate}
				\item $2\ell+3=k$ i.e. $\ell=\frac{k-3}{2}$, we have $T_{\ell}=\{0, 2, \cdots, k-1, k+1\}$, then
				$\{d_{1}+d_{2}: d_{1},d_{2}\in T_{\ell}, d_{1}+d_{2}<n\}=\{0, 2, 3, \cdots, 2k-2, 2k-1,2k, 2k+2 \}=2k+1$.
				\item $2\ell+3\neq k$ i.e. $0<\ell< \frac{k-3}{2}$, we have $T_{\ell}=\{0, 2, \cdots,2\ell+2, 2\ell+4, \cdots, k, k+1\}$, then
				$\{d_{1}+d_{2}: d_{1},d_{2}\in T_{\ell}, d_{1}+d_{2}<n\}=\{0, 2, 3, \cdots, 2k+1,2k+2 \}=2k+2$.
			\end{enumerate}
			Therefore, we have $\dim\left(\mathcal{C}(D, kO, \ell, \eta)^{\star 2}\right)\geq 2k+1$ by Proposition \ref{114}.
			On the other hand, by the relation
			\begin{displaymath}
				ev_{D}\left(\left\langle f\cdot g: f, g\in S_{\ell}\right\rangle\right)= ev_{D}\left(\left\langle \overline{f\cdot g}: f, g\in S_{\ell}\right\rangle\right),  
			\end{displaymath}
			we have $\dim\left(\mathcal{C}(D, kO, \ell, \eta)^{\star 2}\right)= 2k+1$ for $\ell=\frac{k-3}{2}$. Then if $0<\ell<\frac{k-3}{2}$, we have $T_{\ell}=\{0, 2, \cdots,2\ell+2, 2\ell+4, \cdots, k, k+1\}$ and $\dim\left(\mathcal{C}(D, kO, \ell, \eta)^{\star 2}\right)\geq 2k+2$, which is a contradiction.
			
			\item By Theorem \ref{7}, the dual of $\mathcal{C}(D, kO, \ell, \eta)$ is a TGECC with the scalar vector
			\begin{displaymath}
				\mathbf{v}=(\frac{\gamma_{1}}{\beta_{1}},\frac{\gamma_{2}}{\beta_{2}},\cdots,\frac{\gamma_{n}}{\beta_{n}})
			\end{displaymath}
			and defining set $$S_{\ell}^{\perp}=\left\lbrace\sum^{\frac{n-k-1}{2}}_{i=0}a_{i}x^{i}+\sum_{j=0}^{\frac{n-k-5}{2}}b_{i}x^{j}y+b_{\frac{n-k-3}{2}}f^{(1)}_{1}(x,y)\right\rbrace$$, where $f^{(1)}_{1}(x,y)$ is given in Theorem \ref{7}. By the similar discussions as above, we denote by $T_{\ell}^{\perp}=\{0, 2, \cdots, n-k-1\}\cup\{n-2\ell-2\}$ and we have $n-2\ell-2\geq n-k+1$ for $0<\ell\leq\frac{k-3}{2}$. Then we have
			$\{d_{1}+d_{2}: d_{1},d_{2}\in T^{\perp}_{\ell}, d_{1}+d_{2}<n\}=\{0, 2, 3, \cdots, 2n-k-2\ell-3,2n-4\ell-4 \}=2n-k-2\ell-2\geq 2n-2k+1$. By the similar discussions as the first claim, the equality holds if and only if $\ell=\frac{k-3}{2}$.
		\end{enumerate}
	\end{proof}

\section{Conclusion and Future Works}
\label{sec:7}
In this paper, by utilizing the Riemann-Hurtwitz formula and prime ideal decomposition in the elliptic function fields, we initiate the study of twisted elliptic curve codes (TECCs). In particular, we study a class of TECCs with one twist. The parity-check matrices of the TECCs are given by explicitly calculating the Weil differentials. The sufficient and necessary conditions of self-duality are presented. The possible minimum distances of the TECCs are also determined. Moreover, we give some examples of MDS, AMDS, self-dual and MDS self-dual TECCs. On the other hand, we calculate the dimensions of the Schur squares of TECCs and show the non-equivalence between TECCs and ECCs/GRS codes. We list some research problems to conclude the paper:
	\begin{enumerate}
		\item Problem 1:
		In \cite{22}, Hu \textit{et al.} extend the original constructions of TGRS codes to the $(\mathcal{L}, \mathcal{P})-$TGRS codes, therefore it is interesting to extend such constructions to the TECCs.

		\item Problem 2:
		Study more constructions of twisted AG (TAG) codes such as twisted Hermitian codes (THCs), twisted hyper-elliptic curve codes (THECCs), \textit{etc.}.
		
		\item Problem 3:
		Determine the weight distribution of TECCs.

        \item Problem 4:
		Study the combinatorial designs supported by TECCs.
        \end{enumerate}

\end{document}